\documentclass[12pt]{article}
\usepackage{latexsym,amsmath}
\usepackage{graphicx}
\usepackage{caption}
\usepackage{subcaption}
\usepackage{float}
\usepackage{amsfonts}
\usepackage{isomath}
\usepackage{amssymb,array}
\usepackage{epsfig}
\usepackage{color}
\usepackage{setspace}
\usepackage{amsthm}
\usepackage{mathrsfs}
\usepackage[margin=1in]{geometry}
\newtheorem{thm}{Theorem}[section]
\newtheorem{lem}{Lemma}[section]
\newtheorem{re}{Remark}[section]

\newtheorem{p}{Proposition}[section]

\newcommand{\keywords}[1]{\textbf{\textit{Keywords and Phrases: }}#1}
\newcommand{\MSC}[1]{\textbf{\textit{Mathematics Subject Classification--}}#1}

\begin{document}
	
	\title{Estimation of Tsallis entropy and its applications to goodness-of-fit tests}
	\author{$\text{Siddhartha Chakraborty}^{\text{a}}, \text{Asok K. Nanda}^{\text{a}}, \text{Narayanaswamy Balakrishnan}^{\text{b,c}}$\\$^{\text{a}}$ Department of Mathematics and Statistics\\ Indian Institute of Science Education and Research Kolkata\\ West Bengal, India;\\$^{\text{b}}$ Department of Mathematics and Statistics,\\ McMaster University, Hamilton, Ontario, Canada\\$^{\text{c}}$ Department of Mathematics,\\ Atilim University, Ankara, Turkey}
	\date{}
	\maketitle
	\begin{abstract}
		In this paper, we consider the problem of estimating Tsallis entropy from a given data set. We propose four different estimators for Tsallis entropy measure based on higher-order sample spacings, and then discuss estimation of Tsallis divergence measure. We compare the performance of the proposed estimators by means of bias and mean squared error and also examine their robustness to outliers. Next, we propose a spacings-based estimator for Tsallis entropy under progressive type-II censoring and study its performance using Monte Carlo simulations. Another estimator for Tsallis entropy is proposed using quantile function and its consistency and asymptotic normality are studied, and its performance is evaluated through Monte Carlo simulations. Goodness-of-fit tests for normal and exponential distributions as applications are developed using Tsallis divergence measure. The performance of the proposed tests are then compared with some known tests using simulations and it is shown that the proposed tests perform very well. Also, an exponentiality test under progressive type-II censoring is proposed, its performance is compared with an existing entropy-based test using simulation. It is observed that the proposed test performs well. Finally, some real data sets are analysed for illustrative purposes.
	\end{abstract}
	\\
	
	\noindent\MSC{62G05, 62N02, 94A17}
	
\section{Introduction}\label{S1}
Entropy is the foundation of information theory, and it was first formulated by Shannon (1948), independently of Wiener (1948). Entropy for a continuous random variable (rv) $X$ is defined as
\begin{equation}\label{e1}
H(X)=-\int_{A} f(x)\log f(x)dx=E\left[-\log f(X)\right] ,
\end{equation}
where log is the natural logarithm, $f$ is the probability density function (pdf) of $X$, $A$ is the support of $X$, and $E$ is the expectation operator. Entropy is additive in nature, i.e., for two independent rvs $X$ and $Y$, the joint entropy of $X$ and $Y$ is equal to the sum of the individual entropies of $X$ and $Y$; that is, $H(X,Y)=H(X)+H(Y)$.
Apart from information theory, entropy has found key applications in many different fields such as physics, mathematical biology, machine learning, economics and statistics (Cover and Thomas, 2006).\\

\noindent Entropy has been studied extensively in the literature and numerous generalizations of it have been proposed. One of the important generalizations is the R\'enyi entropy, due to R\'enyi (1961). This measure, for a continuous rv $X$, is defined as 

\begin{equation}\label{e2}
H_{\alpha}(X)=\frac{1}{1-\alpha}\log\int_{A}f^{\alpha}(x)dx,\;\;\;0<\alpha\neq 1.
\end{equation}

\noindent Another important generalizations is due to Tsallis (1988), which is given by
\begin{equation}\label{e3}
T_{\alpha}(X)=\frac{1}{\alpha-1}\left( 1-\int_{A}f^{\alpha}(x)dx\right),\;0<\alpha\neq1. 
\end{equation}
R\'enyi and Tsallis entropies are closely related and one can express R\'enyi entropy in terms of Tsallis entropy as
 $$H_{\alpha}(X)=\frac{1}{\alpha-1}\log\left(1-(\alpha-1)T_{\alpha}(X)\right)$$
  if we take $A=[0,\infty)$. The parameter $\alpha$ is called the generalizing parameter which makes R\'enyi and Tsallis entropies more flexible than Shannon entropy. When $\alpha\to 1$, both R\'enyi and Tsallis entropies reduce to the Shannon entropy. When $\alpha>1$, R\'enyi and Tsallis entropies give more weight to high probability events and when $\alpha<1$, they give more weight to rare events. By tuning the generalizing parameter $\alpha$, both R\'enyi and Tsallis entropies can be made to prioritize (penalize) common or rare events, which make them more adventageous than Shannon entropy. Like Shannon entropy, R\'enyi entropy is additive, but Tsallis entropy is non-additive. This non-additivity of Tsallis entropy makes it more useful in studying nonequilibrium and complex systems such as gravity and plasma (see Tsallis, 2009). For more applications of Tsallis entropy in different fields, one may refer to Cartwright (2014), Zhang and Wu (2011), Telesca (2011), and the references therein.\\

\noindent Shannon relative entropy, also known as the Kullback-Leibler (KL) divergence, is widely used as a measure of closeness (dissimilarity) between two distributions. The KL divergence between two random variables $X$ and $Y$, with respective pdfs $f$ and $g$, and having common support $A$, is defined as 

\begin{equation}\label{e4}
K(X;Y)=\int_{A} f(x)\log\frac{f(x)}{g(x)}dx.
\end{equation}

 \noindent Note that $K(X;Y)\geq0$ and equality holds when $X$ and $Y$ have the same distribution. So, a large value of $K(X;Y)$ indicates dissimilarity between the distributions of $X$ and $Y$. Like KL divergence, there exist divergence measures based on R\'enyi and Tsallis entropies. The R\'enyi divergence is defined as 
 \begin{equation}\label{e5}
 H_{\alpha}(X,Y)=\frac{1}{\alpha-1}\log\left( \int_{A} f^{\alpha}(x)g^{1-\alpha}(x)dx\right),\;\;\;\alpha(\neq 1)>0,
 \end{equation}

\noindent and the Tsallis divergence is defined as

\begin{equation}\label{e6}
T_{\alpha}(X;Y)=\frac{1}{\alpha-1}\left(\int_{A} f^{\alpha}(x)g^{1-\alpha}(x)dx-1\right),\;\alpha(\neq 1)>0.
\end{equation}
Like KL divergence, $H_{\alpha}(X,Y)\geq0$ and $T_{\alpha}(X,Y)\geq0$ and equalities hold when $X$ and $Y$ have the same distribution. When $\alpha\to1$, R\'enyi and Tsallis divergence become KL divergence. Both R\'enyi and Tsallis divergence are related to $\alpha$-divergence which, for two continuous rvs $X$ and $Y$ having same support $A$, is defined as $$\alpha(X,Y)=\int_{A} f^{\alpha}(x)g^{1-\alpha}(x)dx.$$ For elaborate discussions on these measures and their generalizations, interested readers may refer to Nielsen \& Nock (2010), Van Erven and Harremos (2014), Li and Turner (2016), and Khammar \& Noughabi (2022).\\

\noindent Estimation of entropy from a data set is an important problem and it has been addressed by many authors over the years. Since the estimation of Shannon entropy by Vasicek (1976), various estimators have been introduced in the literature; see, for example, Van Es (1992), Ebrahimi et al. (1994), Correa (1995), Noughabi (2010), Al-Omari (2014), and Noughabi and Arghami (2022). Although estimation of entropy has gained a lot of attention, estimation of R\'enyi and Tsallis entropies have not been studied to the same extent. Hegde et al. (2005) proposed order statistics based $m$-spacings estimator for R\'enyi entropy and recently, Al-Labadi et al. (2025) introduced an estimator for R\'enyi entropy motivated by the works of Vasicek (1976) and Ebrahimi et al. (1994). Kernel-based estimation of Tsallis entropy has been studied recently for $\rho$-mixing dependent data and length-biased data; see Maya et al. (2024) and Zamini et al. (2024). To the best of our knowledge, order statistics based non-parametric estimators for Tsallis entropy have not been considered in the literature yet. In this paper, we propose some non-parametric estimators for Tsallis entropy using order statistics and compare their performance in terms of bias and mean squarred error. Also, we propose an estimator for Tsallis divergence measure. An estimator for Tsallis entropy under progressive type-II censored data is also developed and its performance is examined through Monte Carlo simulations. Another estimator for Tsallis entropy based on the derivative of the quantile function (qf) of the underlying rv is introduced and its asymptotic properties are studied. Using the estimator for Tsallis divergence, we also develop goodness-of-fit tests for normal and exponential distributions. Goodness-of-fit test for exponential distribution under progressive type-II censored sample is also discussed.\\

\noindent The rest of the paper is organized as follows. We introduce order statistics based non-parametric estimators for Tsallis entropy in Section \ref{S2}. Estimator for Tsallis divergence is studied next in Section \ref{S3}. The performance of these estimators and their robustness to outliers are compared through Monte Carlo simulations in Section \ref{S4}. An estimator for Tsallis entropy under progressive type-II censoring is discussed in Section \ref{S5}. The quantile-based estimator is studied in Section \ref{S6} and its performance is investigated in Section \ref{S7}.  Applications to some real data sets are discussed in Section \ref{S8} wherein goodness-of-fit tests for normal and exponential distributions are explored for complete and progressively type-II censored samples. Finally, some brief concluding remarks are made in Section \ref{S9}.

\section{Estimators for Tsallis entropy based on $m$-spacings}\label{S2}
In this section, we discuss non-parametric estimators for Tsallis entropy measure based on higher order sample spacings. Let $X_1$, $\cdots$,$X_n$ be a random sample of size $n$ drawn from a distribution with cumulative distribution function (cdf) $F$, and let $X_{(1)}\leq X_{(2)}\leq\cdots\leq X_{(n)}$ be the corresponding order statistics. Let $F_n$ be the empirical distribution function of $X$. \\

\noindent Using the transformation $F(x)=w$, entropy in (\ref{e1}) can be expressed as 

\begin{equation}\label{e7}
H(X)=\int_{0}^{1}\log\left(\frac{d}{dw}F^{-1}(w)\right)dw.
\end{equation}

Vasicek (1976) first proposed an estimator for entropy by estimating $\frac{d}{dw}\left(F^{-1}(w) \right)$ from the slope of the line joining the points $(F(X_{(i+m)}),X_{(i+m)})$ and $(F(X_{(i-m)}),X_{(i-m)})$. Here, $m$ is referred as the window size, which is a positive integer less than $\frac{n}{2}$.\\

\noindent Note that $X_{(i)}=X_{(1)}$ if $i<1$ and $X_{(i)}=X_{(n)}$ if $i>n$. Now, by using the empirical estimator, the Vasicek entropy estimator can be presented as

\begin{equation}\label{e8}
\hat{H}(X)=\frac{1}{n}\sum_{i=1}^{n}\log\left(\frac{X_{(i+m)}-X_{(i-m)}}{2m/n}\right).
\end{equation}

\noindent Analogous to Vasicek entropy estimator, we propose an estimator for Tsallis entropy, upon expressing 
\begin{equation}\label{e9}
T_{\alpha}(X)=\frac{1}{\alpha-1}\left(1-\int_{0}^{1}\left( \frac{d}{dw}F^{-1}(w)\right)^{1-\alpha}dw \right), 
\end{equation}

\noindent as
 
\begin{equation}\label{e10}
T_{\alpha}V=\frac{1}{\alpha-1}\left(1-\frac{1}{n}\sum_{i=1}^{n}\left(\frac{X_{(i+m)}-X_{(i-m)}}{2m/n}\right)^{1-\alpha}\right). 
\end{equation}

Here $V$ in $T_{\alpha}V$ emphasizes that this is a Vasicek-type estimator. Note that when $\alpha\to1$, $T_{\alpha}V$ reduces to $\hat{H}(x)$ defined in (\ref{e8}).\\

Hegde et al. (2005) proposed estimators for R\'enyi entropy measure by estimating the density as $$\widehat{f}(x)=\frac{F_n(X_{(i+1)})-F_n(X_{(i)})}{(X_{(i+1)}-X_{(i)})}=\frac{1}{(n+1)(X_{(i+1)}-X_{(i)})}.$$ Using the estimator $\hat{f}(x)$, another estimator for Tsallis entropy can be presented as

\begin{eqnarray}\label{e11}
\hat{T}_{\alpha}(X)&=&\frac{1}{\alpha-1}\left(1-\int_{-\infty}^{+\infty}\hat{f}^{\alpha}(x)dx\right)\nonumber\\
&\approx&\frac{1}{\alpha-1}\left(1-(n+1)^{1-\alpha}\frac{1}{(n-1)}\sum_{i=1}^{n-1}\left(X_{(i+1)}-X_{(i)}\right)^{1-\alpha}\right).
\end{eqnarray}
Learned-Miller and Fisher (2003) noted that single-spacings estimators suffer from a large variance due to the single-interval dependency of the order statistics. This can be avoided by using $m$-spacings estimators. From (\ref{e11}), we can define an $m$-spacings estimator to be

\begin{eqnarray}\label{e12}
T_{\alpha}H&=&\frac{1}{\alpha-1}\left(1-\left( \frac{n+1}{m}\right) ^{1-\alpha}\frac{1}{(n-m)}\sum_{i=1}^{n-m}\left(X_{(i+m)}-X_{(i)}\right)^{1-\alpha}\right)\nonumber\\
&=&\frac{1}{\alpha-1}\left(1-\frac{1}{(n-m)}\sum_{i=1}^{n-m}\left(\frac{X_{(i+m)}-X_{(i)}}{m/(n+1)}\right)^{1-\alpha}\right)
\end{eqnarray}
When $\alpha\to1$, the above estimator reduces to

\begin{equation}\label{e13}
T_1H=\frac{1}{n-m}\sum_{i=1}^{n-m}\log\left(\frac{X_{(i+m)}-X_{(i)}}{m/(n+1)}\right). 
\end{equation}
The estimator in (\ref{e13}) is an estimator for entropy and it is also based on Vasicek entropy estimator in (\ref{e8}).\\

In order to estimate the slope more accurately at the extreme points, Ebrahimi et al. (1994) modified Vasicek entropy estimator as

\begin{equation}\label{e14}
\hat{H}_E(X)=\frac{1}{n}\sum_{i=1}^{n}\log\left(\frac{X_{(i+m)}-X_{(i-m)}}{C_im/n}\right),
\end{equation}

\noindent where $C_i$ is the weight function given by
$$C_i=
\begin{cases}
1+\frac{i-1}{m}\;\;\;\;\mbox{if}\;\;\;1\leq i\leq m,\\
2\;\;\;\;\;\;\;\;\;\;\;\;\;\mbox{if}\;\;\;m+1\leq i \leq n-m,\\
1+\frac{n-i}{m}\;\;\;\;\mbox{if}\;\;\;n-m+1\leq i \leq n.
\end{cases}$$

\noindent Now, upon using this modification, we propose an Ebrahimi-type estimator for Tsallis entropy as

\begin{equation}\label{e15}
T_{\alpha}E=\frac{1}{\alpha-1}\left(1-\frac{1}{n}\sum_{i=1}^{n}\left(\frac{X_{(i+m)}-X_{(i-m)}}{C_im/n}\right)^{1-\alpha}\right). 
\end{equation}

Recently, Noughabi and Arghami (2022) introduced an estimator for entropy by replacing the weights $C_i$ in (\ref{e14}) by a weight function $W_i$, defined as

$$W_i=
\begin{cases}
1,\;\;\;\;\mbox{if}\;\;\;1\leq i\leq m,\\
2,\;\;\;\;\mbox{if}\;\;\;m+1\leq i \leq n-m,\\
1,\;\;\;\;\mbox{if}\;\;\;n-m+1\leq i \leq n.
\end{cases}$$

\noindent This is a modification of Ebrahimi's estimator, where $W_i=\min\limits_i C_i$. They showed that this estimator outperforms the other estimators of entropy. Using the weights $W_i$, yet another estimator for Tsallis entropy can be presented as

\begin{equation}\label{e16}
T_{\alpha}W=\frac{1}{\alpha-1}\left(1-\frac{1}{n}\sum_{i=1}^{n}\left(\frac{X_{(i+m)}-X_{(i-m)}}{W_im/n}\right)^{1-\alpha}\right). 
\end{equation}

\noindent The following theorem states the consistency of the proposed estimators. Proofs follow from the works of Vasicek (1976) and Ebrahimi et al. (1992), and is therefore omitted for brevity.

\begin{thm}
	Let $X_{1},\cdots,X_{n}$ be a random sample from a continuous distribution with cdf $F$ and pdf $f$. Then, as $m,n\to \infty$ and $\frac{m}{n}\to 0$, the estimators $T_{\alpha}V$, $T_{\alpha}H$, $T_{\alpha}E$ and $T_{\alpha}W$ all converge to the true value $T_{\alpha}(X)$.
\end{thm}

\section{Estimation of Tsallis divergence}\label{S3} Suppose $F$ is the cdf of $X$ and $F_{\theta}$ is a parametric distribution function under consideration. Let $f$ and $f_{\theta}$ be the corresponding densities. Then, from (\ref{e6}), Tsallis divergence between $f$ and $f_{\theta}$ can be expressed as

\begin{eqnarray}\label{e17}
T_{\alpha}(F;F_{\theta})=\frac{1}{\alpha-1}\left(\int_{-\infty}^{+\infty} f^{\alpha-1}(x)f^{1-\alpha}_{\theta}(x)dF(x)-1\right). 
\end{eqnarray}
We now propose an estimator for $T_{\alpha}(F;F_{\theta})$ along the lines of Vasicek (1976). Note that Vasicek estimator for $f_n(X_{(i)})$ is $$\hat{f}_n(X_{(i)})=\frac{2m/n}{X_{(i+m)}-X_{(i-m)}}.$$
On using $\hat{f}_n(X_{(i)})$, we can define an estimator for $T_{\alpha}(F;F_{\theta})$ as

\begin{eqnarray}\label{e18}
\widehat{T}_{\alpha}(F;F_{\theta})&=&\frac{1}{\alpha-1}\left(\int_{-\infty}^{+\infty}f_n^{\alpha-1}(x)f_{\theta}^{1-\alpha}(x)dF_n(x)-1\right)\nonumber\\ 
&=&\frac{1}{\alpha-1}\left(\frac{1}{n}\sum_{i=1}^{n}f_n^{\alpha-1}(x)f_{\theta}^{1-\alpha}(x)-1\right)\nonumber\\
&=&\frac{1}{\alpha-1}\left(\frac{1}{n}\sum_{i=1}^{n}\left(\frac{2m}{n(X_{(i+m)}-X_{(i-m)})}\right)^{\alpha-1} f_{\theta}(X_{(i)})^{1-\alpha}-1\right).
\end{eqnarray}

\noindent We will estimate $f_{\theta}(X_{(i)})$ by 

\begin{equation}\label{e19}
f_{\hat{\theta}}(X_{(i)}=\frac{F_{\hat{\theta}}(X_{(i+m)})-F_{\hat{\theta}}(X_{(i-m)})}{X_{(i+m)}-X_{(i-m)}}.
\end{equation}

\noindent Now, from (\ref{e18}) and (\ref{e19}), we finally obtain an estimator for Tsallis divergence to be

\begin{equation}\label{e20}
\widehat{T}_{\alpha}(F,F_{\hat{\theta}})=\frac{1}{\alpha-1}\left(\frac{1}{n}\sum_{i=1}^{n}\left(\frac{2m}{n(F_{\widehat{\theta}}(X_{(i+m)})-F_{\widehat{\theta}}(X_{(i-m)}))} \right)^{\alpha-1}-1\right).
\end{equation}

\noindent We can use $\widehat{T}_{\alpha}(F,F_{\hat{\theta}})$ for model fitting purposes.\\

\noindent Suppose $X_1,\cdots,X_n$ is a random sample from a distribution with cdf $F$. Further, suppose we want to test the hypothesis
\begin{eqnarray}
H_0:F(x)\sim F_{\theta}(x)\;\;\;\;vs.\;\;\;\;H_1:F(x)\not\sim F_{\theta}(x),\nonumber
\end{eqnarray}
 where $F_{\theta}(x)$ is a parametric distribution function whose functional form is known except for the parameter $\theta$. To address this problem, first we will calculate $\widehat{T}_{\alpha}(F,F_{\hat{\theta}})$ using the order statistics and $F_{\widehat{\theta}}(X_{(i)})$ for some choice of $m$ and $\alpha$. Then, larger values of $\widehat{T}_{\alpha}(F,F_{\hat{\theta}})$ will indicate dissimilarity between $F$ and $F_{\theta}$. We will then reject the null hypothesis if $\widehat{T}_{\alpha}(F,F_{\hat{\theta}})>C$, where the critical value C will be determined from the empirical distribution of $\widehat{T}_{\alpha}(F,F_{\hat{\theta}})$.\\
 
 \noindent The following lemma shows that $\widehat{T}_{\alpha}(F,F_{\hat{\theta}})$ is distribution-free.
 
 \begin{lem}
 	The statistic $\widehat{T}_{\alpha}(F,F_{\hat{\theta}})$ is independent of $F_{\hat{\theta}}$.
 \end{lem}
 \begin{proof}
 	Let $X_1,\cdots,X_n$ be a random sample from a continuous distribution having cdf $F_{\hat{\theta}}$. Then, $Z=F_{\hat{\theta}}(X)$ has standard uniform distribution, denoted by U(0,1), with cdf $F_Z(x)=x,\;0<x<1$, which is independent of $\theta$. So, $F_{\hat{\theta}}(X_{(i)}),\;i=1,2,\cdots,n,$ can be treated as order statistics from U(0,1) distribution. Here, the distribution of the statistic $\widehat{T}_{\alpha}(F,F_{\hat{\theta}})$ is independent of $F_{\hat{\theta}}$, as required.
 \end{proof}
 
 \begin{table}[h!]
 	\centering
 	\caption{Bias and MSE of $T_{\alpha}V$, $T_{\alpha}H$, $T_{\alpha}E$ and $T_{\alpha}W$ for N(0,1) distribution.}\label{tab1}
 	\begin{tabular}{c c c c c c c c c c c}
 		\hline
 		& & & \multicolumn{2}{c}{$T_{\alpha}V$}&\multicolumn{2}{c}{$T_{\alpha}H$} & \multicolumn{2}{c}{$T_{\alpha}E$} & \multicolumn{2}{c}{$T_{\alpha}A$}\\
 		\hline
 		$n$ & $m$ & $\alpha$ & Bias & MSE & Bias & MSE & Bias & MSE & Bias & MSE\\
 		\hline
 		20 & 4 & 0.5 & -0.9655 & 1.0257 & -0.8691 & \bf 0.8503 & -1.4840 & 2.2800 & -0.9795 & 1.0753\\
 		& & 1 & -0.3317 & 0.1425 & -0.3130 & 0.1326 & -0.1716 & 0.0621 & -0.0521 & \bf 0.0357\\
 		& & 1.5 & -0.1512 & \bf 0.0355 & -0.1640 & 0.0423 & 0.5228 & 0.2760 & 0.4479 & 0.2043\\
 		& & 2 & -0.0831 & \bf 0.0125 & -0.1084 & 0.0222 & 0.1701 & 0.0296 & 0.1547 & 0.0248\\
 		& & 2.5 & -0.0510 & 0.0064 & -0.0818 & 0.0160 & 0.0717 & 0.0053 & 0.0668 & \bf 0.0047\\
 		& & 3 & -0.0355 & 0.0056 & -0.0684 & 0.0212 & 0.0329 & 0.0012 & 0.0314 & \bf 0.0011\\[1ex]
 		
 		& 6 & 0.5 & -1.0679 & 1.2327 & -0.9110 & 0.9261 & -1.4271 & 2.1136 & -0.6370 & \bf 0.5332\\
 		& & 1 & -0.3761 & 0.1728 & -0.3043 & 0.1258 & -0.1577 & 0.0574 & 0.0370 & \bf 0.0336\\ 
 		& & 1.5 & -0.1732 & 0.0427 & -0.1424 & \bf 0.0341 & 0.4757 & 0.2294 & 0.3448 & 0.1236\\
 		& & 2 & -0.0915 & 0.0137 & -0.0800 & \bf 0.0130 & 0.1582 & 0.0257 & 0.1301 & 0.0179\\
 		& & 2.5 & -0.0536 & 0.0056 & -0.0507 & 0.0062 & 0.0680 & 0.0048 & 0.0597 & \bf 0.0038\\
 		& & 3 & -0.0341 & 0.0028 & -0.0350 & 0.0043 & 0.0320 & 0.0011 & 0.0293 & \bf0.0009\\ [1ex]\hline
 		
 		50 & 7 & 0.5 & -0.6600 & 0.4813 & -0.7673 & \bf 0.6266 & -1.3835 & 1.9506 & -0.8941 & 0.8546\\
 		& & 1 & -0.1657 & 0.0393 & -0.2447 & 0.0719 & -0.0648 & 0.0163 & 0.0292 & \bf 0.0126\\
 		& & 1.5 & -0.0588 & \bf 0.0070 & -0.1146 & 0.0179 & 0.6026 & 0.3637 & 0.5541 & 0.3078\\
 		& & 2 & -0.0266 & \bf 0.0020 & -0.0661 & 0.0066 & 0.1974 & 0.0390 & 0.1896 & 0.0360\\
 		& & 2.5 & -0.0143 & \bf 0.0007 & -0.0426 & 0.0031 & 0.0830 & 0.0069 & 0.0814 & 0.0066\\
 		& & 3 & -0.0088 & \bf 0.0003 & -0.0306 & 0.0017 & 0.0385 & 0.0015 & 0.0381 & 0.0014\\[1ex]
 		
 		 & 16 & 0.5 & -0.8427 & 0.7560 & -0.8938 & 0.8364 & -1.2394 & 1.5772 & -0.2442 & \bf 0.1332\\
 		& & 1 & -0.2383 & 0.0707 & -0.2778 & 0.0913 & -0.0268 & \bf 0.0148 & 0.2069 & 0.0570\\
 		& & 1.5 & -0.0829 & \bf 0.0112 & -0.1137 & 0.0174 & 0.5284 & 0.2803 & 0.3869 & 0.1515\\
 		& & 2 & -0.0325 & \bf 0.0025 & -0.0537 & 0.0046 & 0.1810 & 0.0329 & 0.1540 & 0.0240\\
 		& & 2.5 & -0.0143 & \bf 0.0007 & -0.0277 & 0.0015 & 0.0787 & 0.0062 & 0.0717 & 0.0052\\
 		& & 3 & -0.0065 & \bf 0.0002 & -0.0158 & 0.0006 & 0.0374 & 0.0014 & 0.0353 & 0.0012\\
 		[1ex]\hline
 		
 		100 & 10 & 0.5 & -0.4840 & \bf 0.2614 & -0.6792 & 0.4809 & -1.3323 & 1.7959 & -0.8997 & 0.8425\\
 		& & 1 & -0.0936 & 0.0145 & -0.1995 & 0.0454 & -0.0258 & \bf 0.0064 & -0.1996 & 0.0455\\
 		& & 1.5 & -0.0269 & \bf 0.0023 & -0.0876 & 0.0097 & 0.6327 & 0.4006 & 0.6030 & 0.3639\\
 		& & 2 & -0.0113 & \bf 0.0007 & -0.0482 & 0.0032 & 0.2051 & 0.0421 & 0.2011 & 0.0405\\
 		& & 2.5 & -0.0068 & \bf 0.0003 & -0.0300 & 0.0013 & 0.0856 & 0.0073 & 0.0848 & 0.0072\\
 		& & 3 & -0.0048 & \bf 0.0001 & -0.0201 & 0.0006 & 0.0395 & 0.0016 & 0.0393 & 0.0015\\[1ex]
 		
 		& 33 & 0.5 & -0.7049 & 0.5286 & -0.8957 & 0.8207 & -1.1116 & 1.2641 & -0.0037 & \bf 0.0052\\
 		& & 1 & -0.1600 & 0.0340 & -0.2716 & 0.0800 & 0.0523 & \bf 0.0112 & -0.2710 & 0.0795\\
 		& & 1.5 & -0.0395 & \bf 0.0039 & -0.1058 & 0.0134 & 0.5509 & 0.3040 & 0.4082 & 0.1678\\
 		& & 2 & -0.0077 & \bf 0.0007 & -0.0468 & 0.0030 & 0.1902 & 0.0363 & 0.1644 & 0.0272\\
 		& & 2.5 & 0.0009 & \bf 0.0002 & -0.0223 & 0.0008 & 0.0825 & 0.0068 & 0.0761 & 0.0058\\
 		& & 3 & 0.0023 & \bf 0.00007 & -0.0114 & 0.0002 & 0.0390 & 0.0015 & 0.0372 & 0.0014\\\hline
 	\end{tabular}
 \end{table}
 
\begin{table}[h!]
	\centering
	\caption{Bias and MSE of $T_{\alpha}V$, $T_{\alpha}H$, $T_{\alpha}E$ and $T_{\alpha}W$ for Exp(1) distribution.}\label{tab2}
	\begin{tabular}{c c c c c c c c c c c}
		\hline
		& & & \multicolumn{2}{c}{$T_{\alpha}V$}&\multicolumn{2}{c}{$T_{\alpha}H$} & \multicolumn{2}{c}{$T_{\alpha}E$} & \multicolumn{2}{c}{$T_{\alpha}A$}\\
		\hline
		$n$ & $m$ & $\alpha$ & Bias & MSE & Bias & MSE & Bias & MSE & Bias & MSE\\
		\hline
		20 & 4 & 0.5 & -0.8323 & 0.8662 & -0.7790 & \bf 0.7622 & -1.3041 & 1.8440 & -0.8516 & 0.9307\\
		& & 1 & -0.2585 & 0.1249 & -0.2164 & 0.1081 & -0.1017 & 0.0707 & 0.0219 & \bf 0.0602\\
		& & 1.5 & -0.1785 & 0.0697 & -0.1306 & \bf 0.0578 & 0.6671 & 0.4536 & 0.5635 & 0.3293\\
		& & 2 & -0.1863 & 0.0835 & -0.1236 & \bf 0.0600 & 0.2712 & 0.0811 & 0.2380 & 0.0656\\
		& & 2.5 & -0.2334 & 0.1648 & -0.1503 & 0.0978 & 0.1398 & \bf 0.0296 & 0.1211 & 0.0303\\
		& & 3 & -0.3297 & 0.8022 & -0.1976 & 0.3068 & 0.0652 & \bf 0.0323 & 0.0540 & 0.0450\\[1ex]
		
		& 6 & 0.5 & -0.8964 & 0.9766 & -0.7996 & 0.7946 & -1.2261 & 1.6485 & -0.5142 & \bf 0.5122\\
		& & 1 & -0.2588 & 0.1262 & -0.1837 & 0.0960 & -0.0434 & \bf 0.0627 & 0.1472 & 0.0830\\ 
		& & 1.5 & -0.1566 & 0.0589 & -0.0781 & \bf 0.0423 & 0.6324 & 0.4082 & 0.4665 & 0.2303\\
		& & 2 & -0.1536 & 0.0575 & -0.0537 & \bf 0.0323 & 0.2723 & 0.0787 & 0.2212 & 0.0560\\
		& & 2.5 & -0.1723 & 0.0802 & -0.0513 & 0.0319 & 0.1506 & 0.0275 & 0.1267 & \bf 0.0218\\
		& & 3 & -0.2292 & 0.1990 & -0.0634 & 0.0433 & 0.0881 & 0.0150 & 0.0748 & \bf 0.0142\\[1ex]\hline
		
		50 & 7 & 0.5 & -0.5680 & \bf 0.4127 & -0.6854 & 0.5352 & -1.2225 & 1.5646 & -0.7755 & 0.7158\\
		& & 1 & -0.1287 & 0.0384 & -0.1631 & 0.0496 & -0.0320 & \bf 0.0237 & 0.0686 & 0.0264\\
		& & 1.5 & -0.0862 & \bf 0.0197 & -0.0825 & 0.0202 & 0.7577 & 0.5763 & 0.6800 & 0.4655\\
		& & 2 & -0.0953 & 0.0212 & -0.0671 & \bf 0.0152 & 0.3216 & 0.1047 & 0.2958 & 0.0894\\
		& & 2.5 & -0.1151 & 0.0297 & -0.0694 & \bf 0.0170 & 0.1804 & 0.0336 & 0.1666 & 0.0294\\
		& & 3 & -0.1457 & 0.0613 & -0.0732 & 0.0193 & 0.1137 & 0.0341 & 0.1051 & \bf 0.0127\\[1ex]
		
		 & 16 & 0.5 & -0.6657 & 0.5473 & -0.7800 & 0.6700 & -1.0284 & 1.1480 & -0.1103 & \bf 0.1778\\
		& & 1 & -0.1281 & 0.0430 & -0.1490 & 0.0459 & 0.0808 & \bf 0.0333 & 0.3170 & 0.1275\\
		& & 1.5 & -0.0640 & 0.0150 & -0.0320 & \bf 0.0135 & 0.6904 & 0.4782 & 0.5064 & 0.2608\\
		& & 2 & -0.0680 & 0.0133 & 0.0008 & \bf 0.0076 & 0.3073 & 0.0955 & 0.2544 & 0.0665\\
		& & 2.5 & -0.0822 & 0.0160 & 0.0061 & \bf 0.0053 & 0.1782 & 0.0324 & 0.1561 & 0.0254\\
		& & 3 & -0.0989 & 0.0215 & 0.0055 & \bf 0.0044 & 0.1166 & 0.0140 & 0.1048 & 0.0116\\[1ex]\hline
		
		100 & 10 & 0.5 & -0.4164 & \bf 0.2274 & -0.6159 & 0.4132 & -1.1790 & 1.4322 & -0.7855 & 0.6818\\
		& & 1 & -0.0764 & 0.0169 & -0.1337 & 0.0287 & -0.0065 & \bf 0.0107 & 0.0626 & 0.0148\\
		& & 1.5 & -0.0554 & \bf 0.0089 & -0.0622 & 0.0100 & 0.7941 & 0.6316 & 0.7400 & 0.5488\\
		& & 2 & -0.0652 & 0.0093 & -0.0462 & \bf 0.0068 & 0.3397 & 0.1159 & 0.3215 & 0.1040\\
		& & 2.5 & -0.0762 & 0.0118 & -0.0461 & \bf 0.0066 & 0.1942 & 0.0380 & 0.1840 & 0.0343\\
		& & 3 & -0.0930 & 0.0178 & -0.0488 & \bf 0.0070 & 0.1254 & 0.0160 & 0.1190 & 0.0146\\[1ex]
		
		& 33 & 0.5 & -0.5086 & 0.3359 & -0.7790 & 0.6377 & -0.8878 & 0.8571 & 0.1364 & \bf 0.1479\\
		& & 1 & -0.0594 & \bf 0.0188 & -0.1374 & 0.0306 & 0.1526 & 0.0388 & 0.4005 & 0.1756\\
		& & 1.5 & -0.0254 & \bf 0.0062 & -0.0177 & \bf 0.0062 & 0.7109 & 0.5066 & 0.5237 & 0.2765\\
		& & 2 & -0.0399 & 0.0053 & 0.0137 & \bf 0.0037 & 0.3173 & 0.1012 & 0.2650 & 0.0710\\
		& & 2.5 & -0.0563 & 0.0067 & 0.0212 & \bf 0.0026 & 0.1853 & 0.0346 & 0.1629 & 0.0269\\
		& & 3 & -0.0749 & 0.0100 & 0.0218 & \bf 0.0019 & 0.1220 & 0.0150 & 0.1110 & 0.0126\\\hline
	\end{tabular}
\end{table}

\section{Comparison of estimators}\label{S4} In this section, we compare the performances of the four proposed estimators $T_{\alpha}V$, $T_{\alpha}H$, $T_{\alpha}E$ and $T_{\alpha}W$ by means of Monte Carlo simulations. For this purpose, we generate 10000 random samples from standard normal and standard exponential distributions and calculate the bias and the mean squared error (MSE) of the proposed estimators. Normal distribution with mean $\mu$ and variance $\sigma^2$, denoted by $N(\mu,\sigma^2)$, has the density $$f_N(x;\mu,\sigma)=\frac{1}{\sqrt{2\pi}\sigma}e^{-\frac{1}{2\sigma^2}(x-\mu)^2},\;-\infty<x<+\infty,-\infty<\mu<+\infty,\sigma>0.$$ and exponential distribution with mean $\frac{1}{\lambda}$, denoted by Exp($\lambda$), has the density $$f_E(x;\lambda)=\lambda e^{-\lambda x};\;x>0,\lambda>0.$$ The true value of $T_{\alpha}(X)$ for Normal and Exponential distributions are $\frac{1}{\alpha-1}\left(1-\frac{1}{(2\pi\sigma^2)^{\frac{\alpha-1}{2}}\alpha^{\frac{1}{2}}}\right)$ and $\frac{1}{\alpha-1}\left(1-\frac{\lambda^{\alpha-1}}{\alpha}\right)$, respectively. We calculate the bias and the MSE for the estimators for different choices of $n$, $m$ and $\alpha$. We take $n$ = 20, 50 and 100 and $\alpha$ = 0.5, 1, 1.5, 2, 2.5 and 3. For the choice of $m$, we choose it as $m=[\sqrt{n}+0.5]$, which has been used by Wieczorkowski and Grzegorzewski (1999). Another choice of $m$ is $m=\left[\frac{m}{3}\right]$. So, the choices for $m$ for $n$ = 20, 50 and 100 are (4,7,10) and (6,16,33), respectively. We have presented the obtained results in Tables \ref{tab1} and \ref{tab2}. From these tables, we observe that the performance of the estimators depend on the choice of $m$. As the sample size increases, bias and MSE of the estimators decrease, as expected. From Table \ref{tab1}, we observe that for N(0,1) distribution, $T_{\alpha}W$ outperforms the other estimators when $n$ = 20 for most of the cases. However, $T_{\alpha}V$ performs better than all the estimators in most of the cases for $n$ = 50 and 100. When $\alpha\to1$, i.e., for the KL divergence estimation, $T_{\alpha}W$ performs the best for small sample sizes and $T_{\alpha}E$ performs the best for large sample sizes. Also, in most of the cases, $T_{\alpha}W$ performs better than $T_{\alpha}E$ and that the estimators perform better for $m=\left[\frac{m}{3}\right]$ than for $m=[\sqrt{n}+0.5]$. In general, for N(0,1) distribution, Vasicek-type estimators are better suited with $m=\left[\frac{m}{3}\right]$. From Table \ref{tab2}, it is evident that for Exp(1) distribution, $T_{\alpha}H$ has better performance in most of the cases. However, when the sample size increases, performance of $T_{\alpha}V$ and $T_{\alpha}H$ become similar. Note that $T_{\alpha}E$ has better performance for Exp(1) distribution than N(0,1) distribution. Here, both choices of $m$ are seen to be satisfactory. For $\alpha>1$, $m=\left[\frac{m}{3}\right]$ has overall better performance and the reverse is true for $\alpha$ = 0.5 and 1. From these observations, we can conclude that for practical purposes, Vasicek-type estimator would be quite suitable for both normal and exponential distributions and for both choices of $m$.\\

\subsection{Robustness of the estimators in the presence of outliers} An estimator is said to be sensitive if it gets affected even when there is a small departure from the model assumption. An estimator is robust to outliers if it is stable when the data are subjected to some extreme observations. A non-robust estimator, such as sample mean, can be heavily influenced by outliers. Note that the estimator $T_{\alpha}H$ is similar to $T_{\alpha}V$, and $T_{\alpha}W$ is similar to $T_{\alpha}E$, with a slight modification to the weight function. Here, we study the robustness of $T_{\alpha}V$ and $T_{\alpha}E$. Robustness of $T_{\alpha}H$ and $T_{\alpha}W$ are similar to that of $T_{\alpha}V$ and $T_{\alpha}E$, respectively. The window size $m$ adjusts the smoothing of the estimators and it reduces the variability of the estimators. We discuss the robustness of $T_{\alpha}V$ and $T_{\alpha}E$ when the data are normally distributed, for some choices of $m$, by evaluating their empirical influence function (EIF).\\

\begin{figure}[H]
	\begin{subfigure}{0.55\textwidth}
		\centering
		\includegraphics[width=1\linewidth]{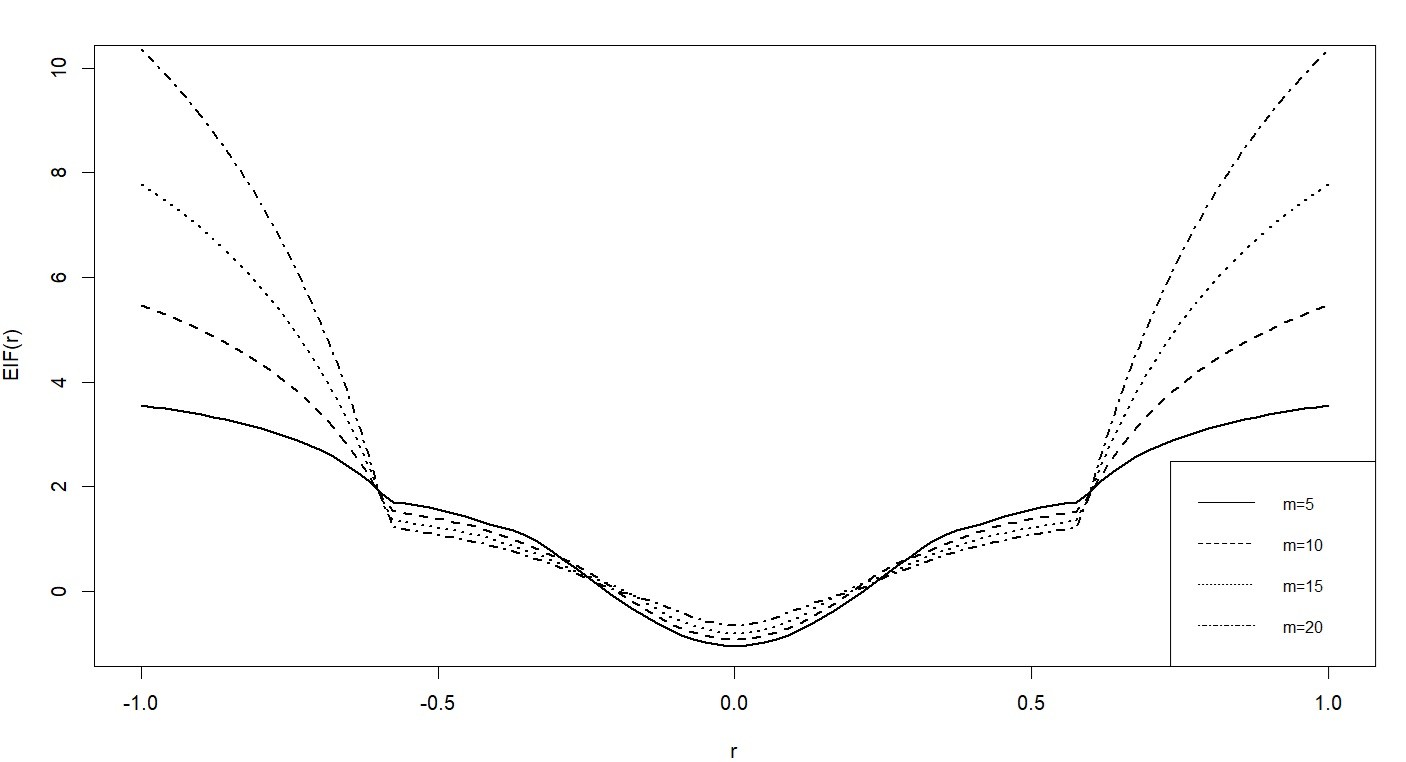}
		\caption{$\alpha$ = 2.}\label{f1.1}
	\end{subfigure}\hfill
	\begin{subfigure}{0.55\textwidth}
		\centering
		\includegraphics[width=1\linewidth]{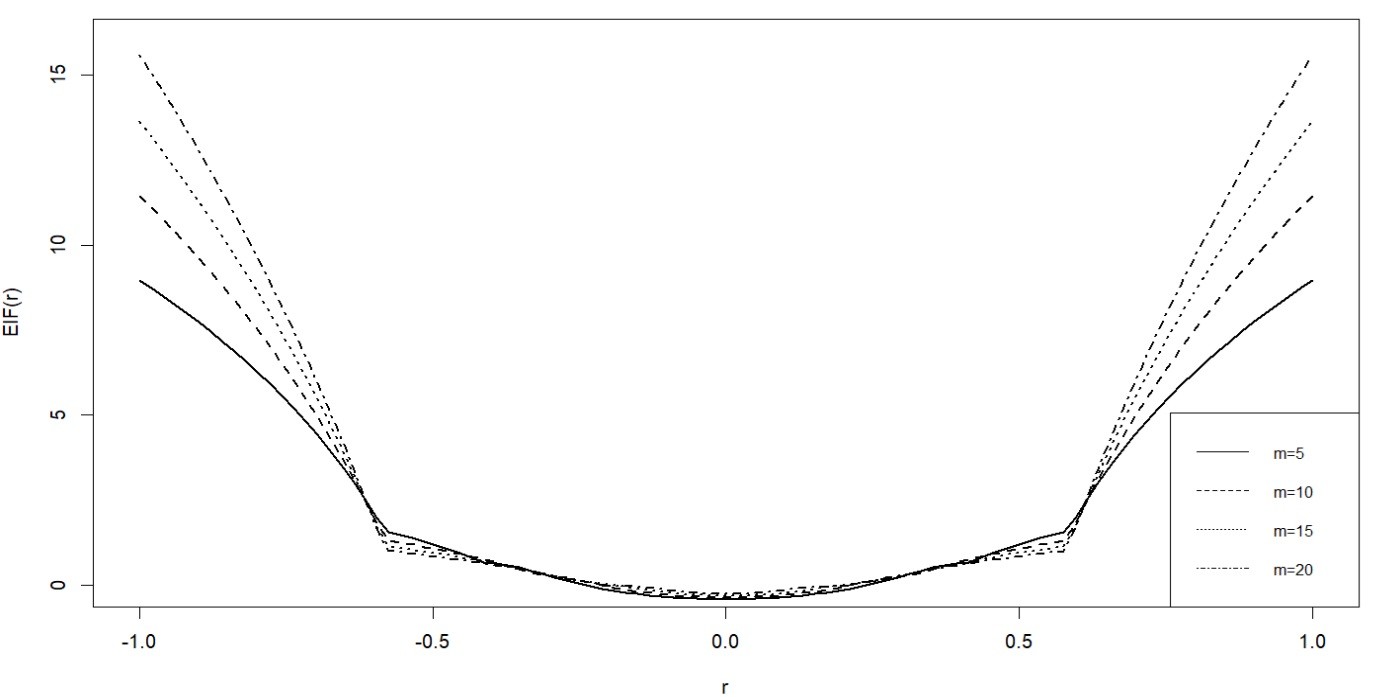}
		\caption{$\alpha$ = 0.9.}\label{f1.2}
	\end{subfigure}
	\caption{Empirical influence function for $T_{\alpha}V$.}
	\label{f1}
\end{figure}

The EIF (Hampel et al., 1986) is used to measure the changes to an estimator subject to outliers. Let $\textbf{X}=(X_1,\cdots,X_n)$ be a random sample and let $\hat{\xi}(\textbf{X})$ be an estimator for $\xi$ based on $\textbf{X}$. Suppose $r$ is a new observation added to the previous sample. EIF($r$) assesses the sensitivity of the estimator $\hat{\xi}$ to $r$ and is defined as

$$EIF(r)=(n+1)(\hat{\xi}(\textbf{X},r)-\hat{\xi}(\textbf{X})).$$

Lower values of EIF indicates more robustness. We calculate EIF for $T_{\alpha}V$ and $T_{\alpha}E$ estimators for $n$ = 100. We generate the samples as $X_i=\mu+\sigma\Phi^{-1}\left( \frac{i}{n+1}\right)$ (see Hampel et al., 1986), choose $r\in[-4\sigma,4\sigma]$ and then calculate the respective EIF values. Here, $\Phi$ is the cdf of the standard normal distribution and so $X_i\sim N(\mu,\sigma^2)$. We take $(\mu,\sigma)$ = (0,0.25) and $m$ = 5, 10, 15 and 20. The obtained results are displayed in Figures \ref{f1} and \ref{f2}, respectively. From Figure (\ref{f1.1}), we observe that for $T_{\alpha}V$ estimator when $\alpha$ = 2, lower values of $m$ perform better for small outliers, while for moderate outliers higher values perform better and for high outliers again lower values perform better. Similar results can be observed from Figure 1(a), for $\alpha$ = 0.9. However, in this case, the difference is negligible for low and moderate outliers. For high outliers, clearly lower choices of $m$ is more robust. For $T_{\alpha}E$ estimator, similar behaviour can be observed for $\alpha$ = 2 from Figure 1(b). For $\alpha$ = 0.9, higher values of $m$ perform better for small to moderate outliers and for higher outliers lower values of $m$ is better. Also, from Figures \ref{f1} and \ref{f2}, it is evident that in general $T_{\alpha}E$ is more robust than $T_{\alpha}V$ since $T_{\alpha}E$ has lower EIF values.\\

\begin{figure}
	\begin{subfigure}{0.55\textwidth}
		\centering
		\includegraphics[width=1\linewidth]{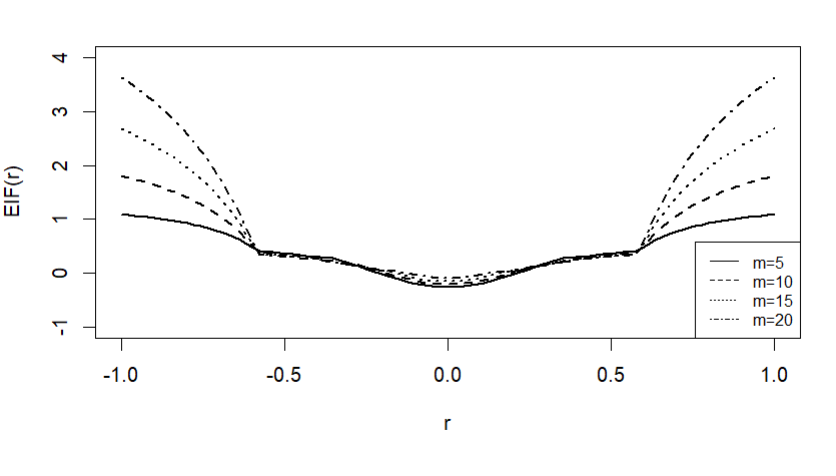}
		\caption{$\alpha$ = 2.}\label{f2.1}
	\end{subfigure}\hfill
	\begin{subfigure}{0.55\textwidth}
		\centering
		\includegraphics[width=1\linewidth]{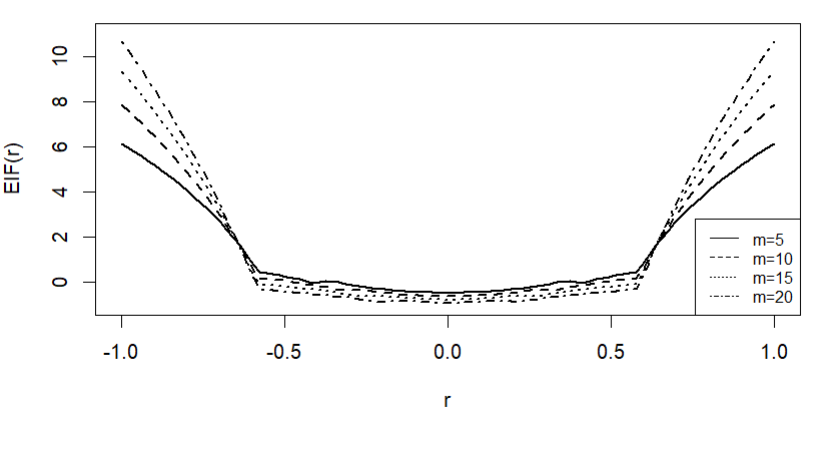}
		\caption{$\alpha$ = 0.9.}\label{f2.2}
	\end{subfigure}
	\caption{Empirical influence function for $T_{\alpha}E$.}
	\label{f2}
\end{figure}

\section{Estimation of Tsallis entropy based on progressive type-II censored data using $m$-spacings}\label{S5}
In this section, we discuss estimation of Tsallis entropy when the data are progressively type-II censored. For this purpose, we consider a Vasicek-type estimator.\\

In reliability and survival analysis, censored life-tests are generally performed as it is not always feasible to conduct experiments till all the units fail. The type-I and type-II censoring are the most common censoring schemes. Suppose $n$ units are placed on a life-testing experiment. In type-I censoring the experiment is continued till a pre-fixed time $T$ while in type-II censoring, the experiment is terminated when a pre-fixed number of failures, say $r(<n)$, is achieved. These censoring schemes do not allow removal of surviving items during the experiment. Cohen (1963) addressed this problem by introducing progressive type-II (PC-II) censoring scheme which allows removal of units during the experiment. A PC-II censoring can be described as follows.\\

 Suppose $n$ identical units are placed on a test, each having a common lifetime distribution with pdf $f$ and cdf $F$. Also, suppose a pre-fixed number $r$ of failures is allowed. Let $R_1,\cdots, R_r$ be prefixed integers such that $R_1 +\cdots+R_r = n-r$. When the first failure occurs, $R_1$ of the remaining $n-1$ surviving units are randomly removed from the test. Then $R_2$ of the remaining $n-R_1-2$ units are randomly removed at the time of second failure. Finally, at the time of the $r$th failure, all the remaining $R_r = n-r-\sum_{i=1}^{r-1}R_i$ units are removed from the test. The failure times thus observed are denoted by $X_{1:r:n},\cdots,X_{r:r:n}$. PC-II censoring reduces to the conventional type-II censoring when $R_i = 0$ for $i=1, \cdots, r-1$, and $R_r=n-r$.\\

For PC-II censored data, we can estimate $\frac{d}{dw}\left(F^{-1}(w) \right)$ as 

\begin{eqnarray}
\frac{d}{dw}\left(F^{-1}(w) \right)&\approx&\frac{X_{i+m:r:n}-X_{i-m:r:n}}{F(X_{i+m:r:n})-F(X_{i-m:r:n})},\nonumber
\end{eqnarray}

\noindent where $m<\frac{r}{2}$, $X_{i:r:n}=X_{1:r:n}$ if $i<1$ and $X_{i:r:n}=X_{r:r:n}$ if $i>r$. So, a Vasicek-type estimator for Tsallis entropy can be defined as

\begin{eqnarray}\label{L1}
T_{\alpha}=\frac{1}{\alpha-1}\left( 1-\frac{1}{r}\sum_{i=1}^{r}\left( \frac{X_{i+m:r:n}-X_{i-m:r:n}}{F(X_{i+m:r:n})-F(X_{i-m:r:n})}\right)^{1-\alpha} \right).
\end{eqnarray}
Let $F(X_{i:r:n})=U_{i:r:n},\;i=1,2,\cdots,r$, be the progressively type-II censored order statistics from Uniform(0,1) distribution. Then,
$$E(U_{i:r:n})=1-\prod_{k=r-i+1}^{r}\beta_k,$$ where $\beta_i=\frac{i+\sum_{k=r-i+1}^{r}R_k}{1+i+\sum_{k=r-i+1}^{r}R_k}$, $\beta_k=\beta_1$ if $k\leq 1$ and $\beta_k=\beta_r$ if $k\geq r$. For detailed discussions on progressively censored ordered statistics, one may see the books by Balakrishnan and Aggarwala (2000) and Balakrishnan and Cramer (2014). We obtain the estimator for Tsallis entropy, by replacing $U_{i:r:n}$ by its expectation in (\ref{L1}), as

\begin{eqnarray}\label{L2}
\hat{T_{\alpha}}=\frac{1}{\alpha-1}\left( 1-\frac{1}{r}\sum_{i=1}^{r}\left( \frac{X_{i+m:r:n}-X_{i-m:r:n}}{\prod_{k=r-(i-m)+1}^{r}\beta_k-\prod_{k=r-(i+m)+1}^{r}\beta_k}\right)^{1-\alpha} \right).
\end{eqnarray}

\begin{p}
	The estimator $\hat{T_{\alpha}}$ is consistent when $r\to n$, $n\to \infty$, $m\to \infty$ and $m/n\to 0$.
\end{p}
\begin{proof}
	When $r\to n$, the PC-II censored sample becomes complete sample and the required result then follows from Vasicek (1976).
\end{proof}

\begin{p}
	For type-II censoring with the first $r$ failures observed, the estimator reduces to $$\hat{T_{\alpha}}=\frac{1}{\alpha-1}\left( 1-\frac{1}{r}\sum_{i=1}^{r}\left( \frac{X_{i+m:r:n}-X_{i-m:r:n}}{2m/(n+1)}\right)^{1-\alpha} \right).$$
\end{p}
\begin{proof}
	The PC-II censoring becomes type-II censoring when $R_i = 0$ for $i=1, \cdots, r-1$, and $R_r=n-r$. In this case, $\beta_j=\frac{j+n-r}{j+1+n-r}$ for all $j$ and $\prod_{k=j}^{r}\beta_k=\frac{j+n-r}{n+1}$. Substituting this in (\ref{L2}), we get the required result.
\end{proof}

\subsection{Empirical evaluation} We perform Monte Carlo simulations to assess the performance of the proposed estimator by examining its mean and variance. We consider different progressive censoring schemes for $n$ = 20 and $r$ = 10 and 15. We choose window size $m=[\sqrt{r}+0.5]$ and $\alpha$ = 1.10, 1.5 and 2. We mostly consider one-step progressive censoring schemes where all the surviving units are removed at one time during the experiment. We simulate 5000 progressively type-II censored samples from each of Exp(1) and N(0,1) distributions and calculate the average estimate (AE) and the variance of $\hat{T_{\alpha}}$. The obtained results are presented in Tables \ref{TN1} and \ref{TN2}, respectively. In these tables, the notation $a*b$ refers to `$a$ is repeated $b$ times'. For instance, $(0*5,10,0*4)$ refers to $(0,0,0,0,0,10,0,0,0,0)$.

\begin{table}[h!]
	\centering
	\caption{Average Estimate and Variance of $\hat{T_{\alpha}}$ for Exp(1) distribution.}\label{TN1}
	\begin{tabular}{c c c c c c c c c}
		\hline
		& & & \multicolumn{2}{c}{$\alpha=1.10$}&\multicolumn{2}{c}{$\alpha=1.50$} & \multicolumn{2}{c}{$\alpha=2$}\\
		\hline
		$n$ & $r$ & Schemes & AE & Var & AE & Var & AE & Var\\
		\hline
		20 & 10 & ($10,0*9$) & -0.4850 & 0.1413 & -0.5991 & 0.2534 & -0.8415 & 0.6028\\
		& & ($0,10,0*8$) & -0.4772 & 0.1354 & -0.5915 & 0.2328 & -0.7865 & 0.5766\\
		& & ($0*4,10,0*5$) & -0.3658 & 0.1324 & -0.4783 & 0.2296 & -0.7081 & 0.5138\\
		& & ($0*5,10,0*4$) & -0.3246 & 0.1217 & -0.4302 & 0.2147 & -0.6118 & 0.3846\\
		& & ($0*8,10,0$) & -0.0850 & 0.1208 & -0.1413 & 0.1530 & -0.2716 & 0.2809\\
      	& & ($0*9,10$)   & \bf 0.2153 & \bf0.1188 & \bf0.1621 & \bf0.1189 & \bf0.0906 & \bf0.1392\\
    	& & ($5,0*8,5$) & -0.0578 & 0.1300 & -0.1215 & 0.1648 & -0.2237 & 0.3146\\
	    & & ($1*10$) & -0.2000 & 0.1300 & -0.2807 & 0.1683 & -0.4435 & 0.4890\\\hline
    	20 & 15 & ($5,0*14$) & 0.2097 & 0.0796 & 0.1420 & 0.0737 & 0.0762 & 0.0860\\
    	& & ($0,5,0*13$) & 0.2243 & 0.0734 & 0.1550 & 0.0712 & 0.0832 & 0.0772\\
    	& & ($0*7,5,0*7$)& 0.2451 & 0.0724 & 0.1899 & 0.0670 & 0.1473 & 0.0656\\
    	& & ($0*13,5,0$) & 0.3602 & 0.0695 & 0.2955 & 0.0566 & 0.2348 & 0.0595\\
    	& & ($0*14,5$)   & \bf0.4884 & \bf0.0642 & \bf0.4074 & \bf0.0534 & \bf0.2941 & \bf0.0493\\
    	& & ($3,0*13,2$) & 0.3289 & 0.0708 & 0.2526 & 0.0665 & 0.1900 & 0.0641\\\hline
			
	\end{tabular}
\end{table}

\begin{table}[h!]
	\centering
	\caption{Average Estimate and Variance of $\hat{T_{\alpha}}$ for N(0,1) distribution.}\label{TN2}
	\begin{tabular}{c c c c c c c c c}
		\hline
		& & & \multicolumn{2}{c}{$\alpha=1.10$}&\multicolumn{2}{c}{$\alpha=1.50$} & \multicolumn{2}{c}{$\alpha=2$}\\
		\hline
		$n$ & $r$ & Schemes & AE & Var & AE & Var & AE & Var\\
		\hline
		20 & 10 & ($10,0*9$) & 0.5466 & 0.0815 & 0.4286 & 0.0610 & 0.3332 & 0.0696\\
		& & ($0,10,0*8$) & 0.5509 & 0.0781 & 0.4301 & 0.0565 & 0.3386 & 0.0557\\
		& & ($0*4,10,0*5$) & 0.6267 & 0.0760 & 0.4507 & 0.0677 & 0.2964 & 0.0742\\
		& & ($0*5,10,0*4$) & 0.6655 & 0.0739 & 0.4866 & 0.0646 & 0.2992 & 0.0734\\
		& & ($0*8,10,0$) & 0.8831 & 0.0696 & 0.6517 & 0.0496 & 0.4594 & 0.0511\\
		& & ($0*9,10$)   & \bf1.2028 & \bf0.0587 & \bf0.8926 & \bf0.0331 & \bf0.6707 & \bf0.0131\\
		& & ($5,0*8,5$) & 0.9141 & 0.0690 & 0.7283 & 0.0443 &  0.5538 & 0.0287\\
		& & ($1*10$) & 0.7937 & 0.0716 & 0.5973 & 0.0533 & 0.4384 & 0.0530\\\hline
		20 & 15 & ($5,0*14$) & 0.8832 & 0.0394 & 0.6824 & 0.0222 & 0.5343 & 0.0137\\
		& & ($0,5,0*13$) & 0.8411 & 0.0389 & 0.6844 & 0.0222 & 0.5397 & 0.0133\\
		& & ($0*7,5,0*7$)& 0.8753 & 0.0375 & 0.6905 & 0.0216 & 0.5235 & 0.0138\\
		& & ($0*13,5,0$) & 0.9964 & 0.0372 & 0.7746 & 0.0189 & 0.5871 & 0.0126\\
		& & ($0*14,5$)   & \bf1.0980 & \bf0.0362 & \bf0.8625 & \bf0.0173 & \bf0.6543 & \bf0.0076\\
		& & ($3,0*13,2$) & 0.9558 & 0.0375 & 0.7593 & 0.0209 & 0.5903 & 0.0116\\\hline
		
	\end{tabular}
\end{table}

From these tables, we observe that when the sample size increases, the variance of the proposed estimator decreases. Also, when the average value of the proposed estimator increases, the variance of the estimator decreases. This is due to the fact that Tsallis entropy works as an information measure when $\alpha>1$. So a large value of Tsallis entropy indicates that the information of the corresponding experiment is high, which results in smaller variance. Among all progressive censoring schemes, type-II censoring has the maximum information. It can be observed that for type-II censoring schemes, the average estimate is the highest, i.e., the Tsallis entropy is maximum for type-II censoring schemes and consequently, the variance is the lowest. Furthermore, the performance of the proposed estimator is better for N(0,1) than for Exp(1) distribution.
 
\section{Estimation of Tsallis entropy through quantile approach}\label{S6} Both cdf and qf express the same information about the rv. `Sum of two qfs is also a qf' is an important property, which cdf does not possess; see Unnikrishnan Nair et al. (2014). Suppose $Q$ is the qf of a rv $X$ defined as

\begin{eqnarray}
Q(w)=F^{-1}(w)= \mbox{inf}\{x:F(x)\geq w\},\;0\leq w\leq 1.\nonumber
\end{eqnarray}
The quantile density function (qdf) is defined as $q(w)=\frac{d}{dw}Q(w)=\frac{1}{f(Q(w))}$. Then, we can express Tsallis entropy in terms of qdf as follows:

\begin{eqnarray}\label{e21}
T_{\alpha}(X)&=&\frac{1}{\alpha-1}\left(1-\int_{0}^{1}\left( q(w)\right)^{1-\alpha}dw \right)\nonumber\\
&=& \frac{1}{\alpha-1}\left(1-\int_{0}^{1}\left( f(Q(w))\right)^{\alpha-1}dw \right).
\end{eqnarray}
Then, following the work of Parzen (1979), an empirical estimator of qf can be given as

\begin{eqnarray}
Q_n(w)=X_{(i)},\,\,\,\frac{i-1}{n}<w\leq\frac{i}{n},\;\;\;i=1,2,\cdots,n.\nonumber
\end{eqnarray}
Parzen (1979) also provided a smooth estimator of qf to be
\begin{eqnarray}
\bar{Q}_n(w) = n\left(\frac{i}{n} -w\right)X_{(i-1)} + n\left(w- \frac{i-1}{n}\right)X_{(i)}, \ \mbox{for} \ \frac{i-1}{n} \leq w \leq \frac{i}{n},\nonumber
\end{eqnarray} 
where $ i= 1,\cdots,n$ and $X_{(0)} =0$. In order to estimate $T_{\alpha}(X)$, we first need to estimate $q(w)$. Jones (1992) proposed an estimator for qdf as
\begin{eqnarray}\label{e22}
\hat{q}(w)=\frac{1}{\hat{f}_n(Q_n(w))},
\end{eqnarray}
 where $\hat{f}_n$ is the kernel density estimator of $f$ defined as $$\hat{f}_n(x)=\frac{1}{nh_n}\sum_{i=1}^{n}K\left(\frac{x-X_i}{h_n}\right);$$ here $h_n$ is the bandwidth parameter and the kernel function $K$ satisfies the following properties:\\
1. $K(w)\geq0$ for all $w$;\\
2. $\int_{-\infty}^{+\infty}K(w)dw=1$;\\
3. $K(\cdot)$ is symmetric about zero;\\
4. $K(\cdot)$ satisfies the Lipschitz condition, i.e., there exists a positive constant $M$ such that $\mid K(v)-K(w)\mid\leq M\mid v-w\mid$, for all $u,v$.\\

\noindent Thus, an estimator of $T_{\alpha}(X)$ can be defined as

\begin{eqnarray}\label{e23}
T^q_{\alpha}(X)= \frac{1}{\alpha-1}\left(1-\int_{0}^{1}\left( \hat{f}_n(Q_n(w))\right)^{\alpha-1}dw \right).
\end{eqnarray}

\noindent We now examine the asymptotic properties of $T^q_{\alpha}(X)$. We first show that $T^q_{\alpha}(X)$ is a consistent estimator of $T_{\alpha}(X)$.

\begin{thm}\label{th1}
	Let $X$ be a rv with pdf $f(\cdot)$ and qdf $q(\cdot)$. Suppose $\hat{q}(w)$ is the estimator of $q(w)$ defined in (\ref{e22}). Then, $T^q_{\alpha}(X)$ is consistent for $T_{\alpha}(X)$.
\end{thm}

\begin{proof}
	From (\ref{e23}), we have
	
	\begin{eqnarray}\label{e24}
	T^q_{\alpha}(X)-T_{\alpha}(X)=-\frac{1}{\alpha-1}\left( \int_{0}^{1}\left( \hat{f}^{\alpha-1}_n(Q_n(w))-f^{\alpha-1}(Q(w))\right)dw\right).
	\end{eqnarray}

\noindent Soni et al. (2012) proved that $\hat{f}_n(Q_n(w))$ is consistent for $f(Q(w))$, i.e., as $n\to \infty$, $\hat{f}_n(Q_n(w))-f(Q(w))\to 0$. Then, from continuous mapping theorem, we have\\ $\hat{f}^{\alpha-1}_n(Q_n(w))-f^{\alpha-1}(Q(w))\to 0$ as $n\to \infty$ and so from (\ref{e24}), we can write $T^q_{\alpha}(X)-T_{\alpha}(X)\to 0$ as $n\to \infty$. Hence, the theorem.
\end{proof}

\noindent The following theorem shows that $T^q_{\alpha}(X)$ is asymptotically normal.

\begin{thm}\label{th2}
	Suppose $X$ is a rv having Tsallis entropy $T_{\alpha}(X)$ and $T^q_{\alpha}(X)$ is the estimator of $T_{\alpha}(X)$ defined in (\ref{e23}). Then, $\sqrt{n}\left( T^q_{\alpha}(X)-T_{\alpha}(X)\right) $ is asymptotically normally distributed with mean zero and variance $$\sigma^2=\frac{n}{(\alpha-1)^2}E\left( \int_{0}^{1}\hat{f}^{\alpha-1}_n(Q_n(w))dw\right)^2.$$
\end{thm}

\begin{proof}
	From (\ref{e24}), we get
	
	\begin{eqnarray}\label{e25}
	\sqrt{n}\left( T^q_{\alpha}(X)-T_{\alpha}(X)\right) =-\frac{\sqrt{n}}{\alpha-1}\left( \int_{0}^{1}\left( \hat{f}^{\alpha-1}_n(Q_n(w))-f^{\alpha-1}(Q(w))\right)dw\right).
	\end{eqnarray}
	Soni et al. (2012) showed that $\sqrt{n}\left( \hat{f}_n(Q_n(w))-f(Q(w))\right)$ is asymptotically normal with zero mean and variance $\tau^2$, say. Then, by using delta method, we can easily show that $$\sqrt{n}\left( \hat{f}^{\alpha-1}_n(Q_n(w))-f^{\alpha-1}(Q(w))\right)\sim N(0,\left((\alpha-1)f^{\alpha-2}(Q(w))\right)^2\tau^2).$$ 
	As $\sqrt{n}\left( \hat{f}^{\alpha-1}_n(Q_n(w))-f^{\alpha-1}(Q(w))\right)$ is asymptotically normal with zero mean, by applying functional delta method and Slutsky's theorem to (\ref{e25}), we then obtain the desired result.
	
\end{proof}

\begin{re}\label{R1}
Soni et al. (2012) proposed an estimator for qdf as 
\begin{eqnarray}
\tilde{q}_n(w)=\frac{1}{h_n}\int_{0}^{1}\frac{K\left( \frac{z-w}{h_n}\right) }{\hat{f}_n(Q_n(w))}dz,\nonumber
\end{eqnarray}
where $K$ is the kernel function and $h_n$ is the bandwidth parameter. This estimator is consistent and asymptotically normal. Using $\tilde{q}_n(w)$, we can similarly define another estimator of Tsallis entropy as 
$$\tilde{T}_{\alpha}(X)=\frac{1}{\alpha-1}\left(1-\int_{0}^{1}\left(\tilde{q}(w)\right)^{1-\alpha}dw \right).$$ Proceeding along the same lines as in Theorems \ref{th1} and \ref{th2}, we can establish the consistency and asymptotic normality of $\tilde{T}_{\alpha}(X)$ as well.
\end{re}

\section{Empirical study}\label{S7} In this section, we perform a simulation study to assess the performance of the estimator $T^q_{\alpha}(X)$. For reference distributions, we consider standard normal, standard exponential and Govindarajulu distributions. Govindarajulu distribution is an important distribution to model lifetime data, and it has the qf given by
$$Q(w)=\mu+\sigma\left((\gamma+1)w^{\gamma}-\gamma w^{\gamma+1}\right),\;\;\mu,\sigma,\gamma>0,\;\;0\leq w \leq 1,$$
and is denoted by Govindarajulu ($\mu,\sigma,\gamma$). We generate 5000 random samples from N(0,1) and Exp(1) distributions with sample size $n$ = 20, 50 and 100. We consider $\alpha$ = 0.5, 1.5, 2, 2.5 and 3 and then calculate the bias and the MSE of $T^q_{\alpha}(X)$. We use the Gaussian kernel and select bandwidth according to the normal reference rule $h_n=1.06\sigma n^{-\frac{1}{5}}$ (cf. Silverman, 2018). The obtained results are provided in Table \ref{T2.1}. From this table, we observe that as the sample size increases, the bias and the MSE of $T^q_{\alpha}(X)$ for both distributions decrease. The estimator performs better when $\alpha$ increases. Also, $T^q_{\alpha}(X)$ performs better for N(0,1) than Exp(1) distribution. This may be due to the fact that we have chosen Gaussian Kernel and that the choice of $h_n$ is appropriate when data come from normal distribution.\\ 

\begin{table}[H]
	\centering
	\caption{Bias and MSEs of $T^q_{\alpha}(X)$ for N(0,1) and Exp(1) distributions.}\label{T2.1}
	\begin{tabular}{p{1cm} p{1cm} p{2cm} p{2cm} p{2cm} p{1cm}}
		\hline
		$n$ & $\alpha$ &  \multicolumn{2}{c} {N(0,1)} & \multicolumn{2}{c}{Exp(1)}  \\ \cline{3-6}
	   &      &  Bias    & MSE      &  Bias    &  MSE   \\ \hline
	20 & 0.5  & -0.4516  & 0.3130   & -0.2156  & 0.3142\\
	   & 1.5  & -0.0200  & 0.0084   &  0.1582  & 0.0497\\
	   &  2   & -0.0047  & 0.0028   &  0.1361  & 0.0275\\
	   & 2.5  & -0.0006  & 0.0008   &  0.1110  & 0.0159\\
	   &  3   & -0.00005 & 0.0003   &  0.0874  & 0.0096\\
   50  & 0.5  & -0.3367  & 0.1618   & -0.1497  & 0.1223\\
       & 1.5  & -0.0014  & 0.0029   &  0.1547  & 0.0328\\
       &  2   &  0.0030  & 0.0010   & 0.1359   & 0.0214\\
       & 2.5  &  0.0034  & 0.0003   & 0.1075   & 0.0129\\
       &  3   &  0.0023  & 0.0001   & 0.0859   & 0.0081\\
  100  & 0.5  & -0.2588  & 0.0917   & -0.1007  & 0.0676\\
       & 1.5  & -0.0014  & 0.0014   & 0.1488   & 0.0263\\
       &  2   &  0.0044  & 0.0005   & 0.1223   & 0.0165\\
       & 2.5  &  0.0038  & 0.0002   & 0.1030   & 0.0113\\
       &  3   &  0.0019  & 0.00006  & 0.0846   & 0.0075\\       
		
		\hline
	\end{tabular}
\end{table}

To check the asymptotic normality of $T^q_{\alpha}(X)$, we generate 1000 samples of size 100 and 200 from N(0,1) and Exp(1) and calculate ($T^q_{\alpha}(X)-T_{\alpha}(X)$) for $\alpha$ = 0.5 and 1.5. Now, ($T^q_{\alpha}(X)-T_{\alpha}(X)$) is asymptotically normal with zero mean. We plot the histograms of ($T^q_{\alpha}(X)-T_{\alpha}(X)$) in Figures \ref{fig3} and \ref{fig4}, respectively. From these figures, we see that the asymptotic distributions are indeed normal for both cases. Also, we check the normality for the simulated data set of ($T^q_{\alpha}(X)-T_{\alpha}(X)$) for each scenario using Kolmogorov-Smirnov test and every time the test passed with large p-values. Although the bandwidth is chosen based on normal reference rule, it still provides good results even when the underlying distribution is exponential. Other choices for the bandwidth parameter can also be explored when the underlying data are not normal. 
\begin{figure}
	\begin{subfigure}{0.5\textwidth}
		\centering
		\includegraphics[width=1\linewidth]{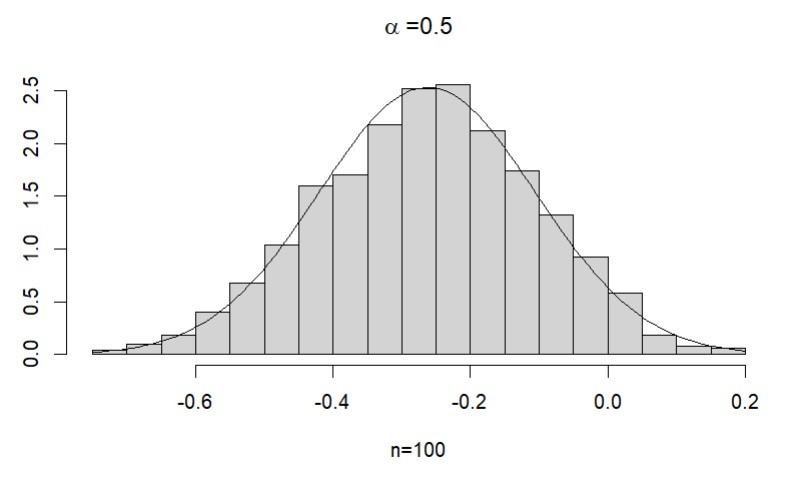}
	\end{subfigure}\hfill
	\begin{subfigure}{0.5\textwidth}
		\centering
		\includegraphics[width=1\linewidth]{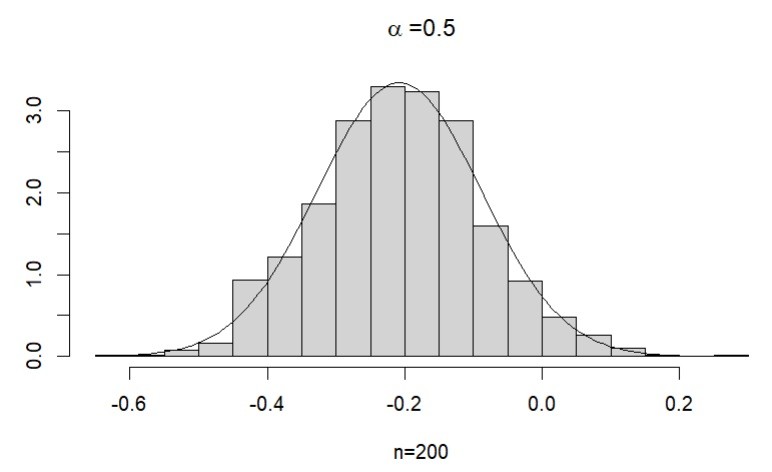}
	\end{subfigure}
	\begin{subfigure}{0.5\textwidth}
		\centering
		\includegraphics[width=1\linewidth]{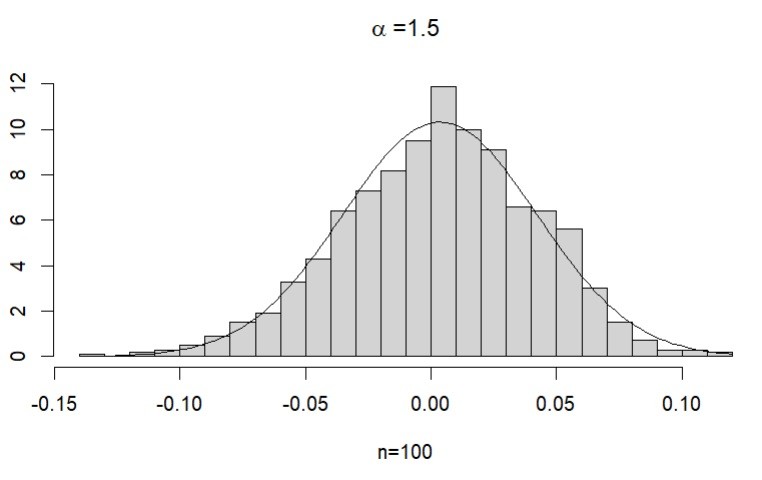}
	\end{subfigure}\hfill
	\begin{subfigure}{0.5\textwidth}
		\centering
		\includegraphics[width=1\linewidth]{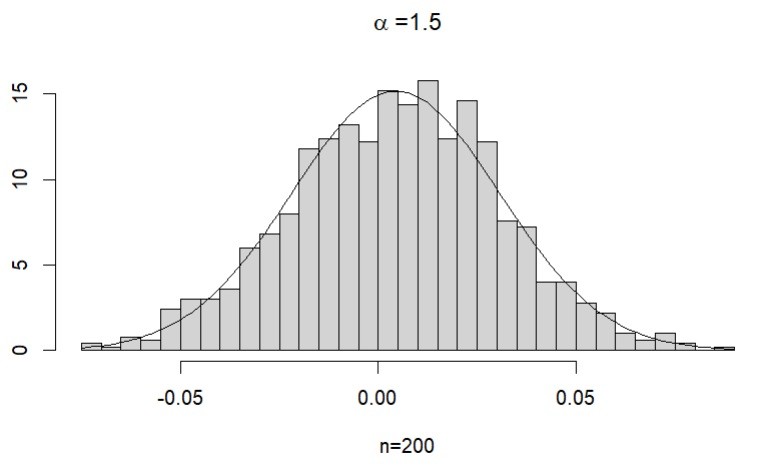}
	\end{subfigure}
	\caption{Histogram of $T^q_{\alpha}(X)-T_{\alpha}(X)$ for N(0,1).}
	\label{fig3}
\end{figure}

\begin{figure}
	\begin{subfigure}{0.5\textwidth}
		\centering
		\includegraphics[width=1\linewidth]{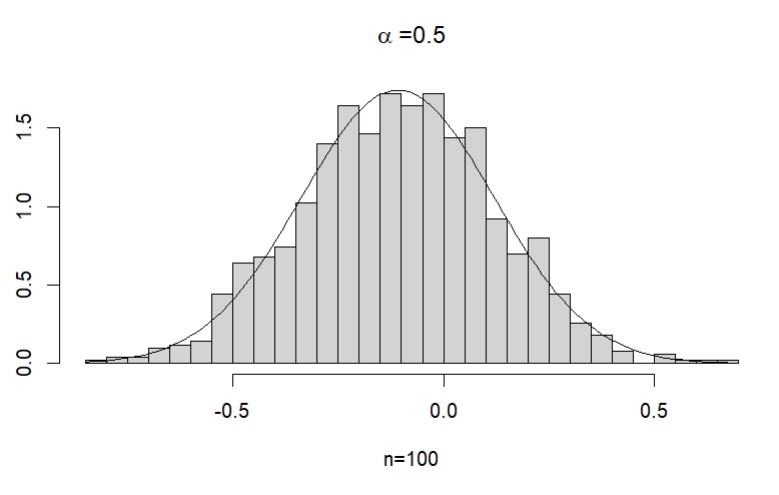}
	\end{subfigure}\hfill
	\begin{subfigure}{0.5\textwidth}
		\centering
		\includegraphics[width=1\linewidth]{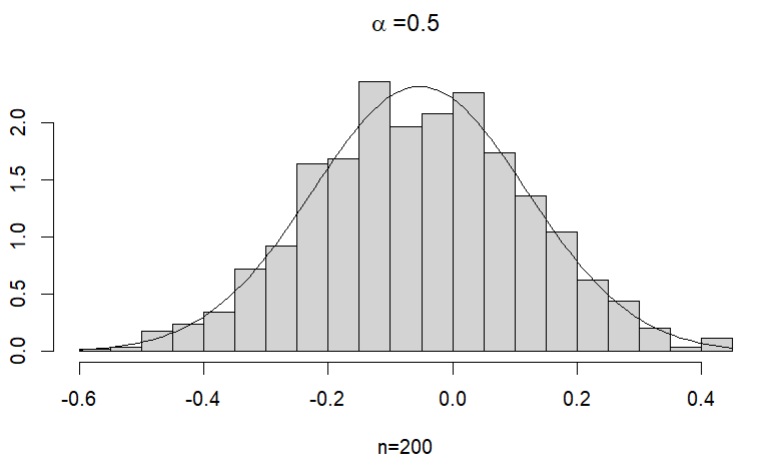}
	\end{subfigure}
	\begin{subfigure}{0.5\textwidth}
		\centering
		\includegraphics[width=1\linewidth]{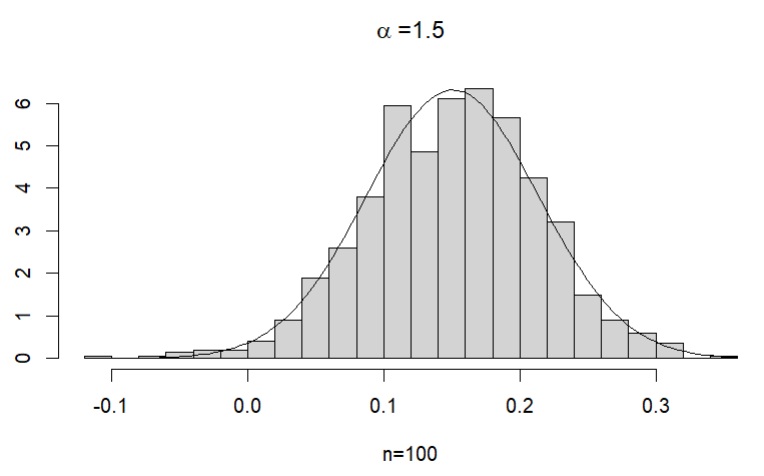}
	\end{subfigure}\hfill
	\begin{subfigure}{0.5\textwidth}
		\centering
		\includegraphics[width=1\linewidth]{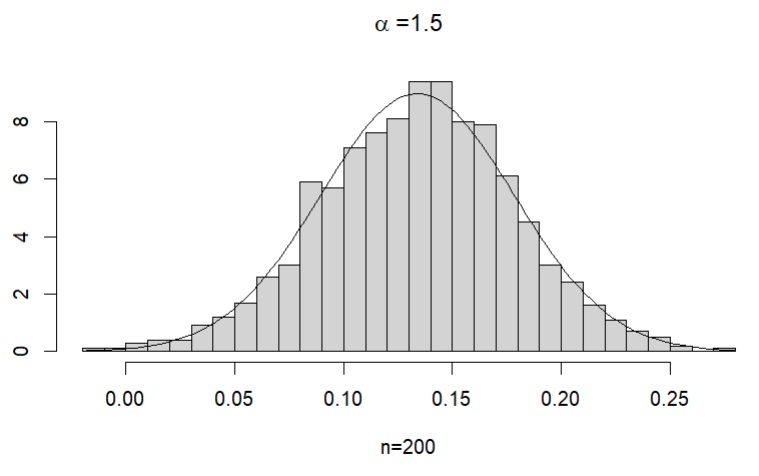}
	\end{subfigure}
	\caption{Histogram of $T^q_{\alpha}(X)-T_{\alpha}(X)$ for Exp(1).}
	\label{fig4}
\end{figure}

There are some distributions that do not have tractable cdf, but have simple and closed form of qf, such as power-Pareto, skew lambda and Govindarajulu distributions. In these situations, it is advantageous to use $T^q_{\alpha}(X)$ to estimate Tsallis entropy. 
We simulate 5000 random samples from Govindarajulu distribution with parameters (0,0.75,0.25) and (0,0.25,0.75), for $n$ = 50, 100 and 200, and then calculate the bias and the MSE of $T^q_{\alpha}(X)$. We chose $\alpha$ = 0.15 and 0.75. Note that, for these models, Tsallis entropy does not exist for $\alpha>1$. We provide the obtained results in Table \ref{T2.2}. The true value of Tsallis entropy are also included in the table. From this table, we observe that the estimator performs well for these models. Both Bias and MSE decrease as sample size increases and the performance of the estimator seems to be better when $\alpha$ is small.

\begin{table}[H]
	\centering
	\caption{Bias and MSEs of $T^q_{\alpha}(X)$ for Govindarajulu distribution.}\label{T2.2}
	\begin{tabular}{p{2.5cm} p{1cm} p{2cm} p{1cm} p{2cm} p{1cm}}
		\hline
	$(\mu,\sigma,\gamma)$ & $\alpha$ &  True value & $n$ &  Bias &  MSE \\\hline
		(0,0.75,0.25) & 0.15 & -0.4592 &  50 & -0.1298 & 0.0242\\
		              &      &         & 100 & -0.1033 & 0.0148\\
		              &      &         & 200 & -0.0784 & 0.0087\\
		              & 0.75 & -1.1769 &  50 & 0.2353  & 0.0733\\
		              &      &         & 100 & 0.2299  & 0.0608\\
		              &      &         & 200 & 0.2245  & 0.0542\\
		(0,0.25,0.75) & 0.15 & -0.8308 &  50 & 0.0357  & 0.0018\\
		              &      &         & 100 & 0.0351  & 0.0014\\
		              &      &         & 200 & 0.0335  & 0.0012\\
		              & 0.75 & -1.3919 &  50 & 0.1921  & 0.0416\\
		              &      &         & 100 & 0.1836  & 0.0354\\
		              &      &         & 200 & 0.1730  & 0.0308\\             
		\hline
	\end{tabular}
\end{table}

\section{Applications to Goodness-of-fit testing}\label{S8} In this section, we develop goodness-of-fit tests for normal and exponential distributions. We use $\widehat{T}_{\alpha}(F,F_{\theta})=\frac{1}{\alpha-1}\left(\frac{1}{n}\sum_{i=1}^{n}\left(\frac{2m}{n(F_{\widehat{\theta}}(X_{(i+m)})-F_{\widehat{\theta}}(X_{(i-m)}))} \right)^{\alpha-1}-1\right)$ as the test statistic and choose $m$ as $m=\sqrt{n}+0.5$.

\begin{table}[h!]
	\centering
	\caption{Power of the tests for N(0,1) distribution.}\label{tab3}
	\begin{tabular}{|c| c| c| c c c| c| c c c|}
		\hline
		Alternatives & $\alpha$ & $n(m)$ & $\widehat{T}_{\alpha}(F,F_{\theta})$ & $KL(F,F_{\hat{\theta}})$ & KS & $n(m)$ & $\widehat{T}_{\alpha}(F,F_{\theta})$ & $KL(F,F_{\hat{\theta}})$ & KS\\
		\hline
		N(0,1)& 0.5 & 10(3)   & 0.054 & 0.055 & 0.047 & 15(4) & 0.052 & 0.051 & 0.048\\ 
		& 2   &     & 0.049 & 0.055 & 0.047 & & 0.051 & 0.051 & 0.048\\
		& 3   &     & 0.053 & 0.055 & 0.047 & & 0.054 & 0.051 & 0.048\\[1ex]
		C(0,1)& 0.5 &   & \bf0.545 & 0.481 & 0.489 & & \bf0.697 & 0.613 & 0.685\\
		& 2   &   & 0.373 & 0.481 & \bf0.489 & & 0.510 & 0.613 & \bf0.685\\
		& 3   &   & 0.320 & 0.481 & \bf0.489 & & 0.454 & 0.613 & \bf0.685\\[1ex]
		Exp(1)& 0.5 &  & \bf0.489 & 0.474 & 0.418 & & \bf0.725 & 0.712 & 0.562\\
		& 2   &   & 0.409 & \bf0.474 & 0.418 & & 0.666 & \bf0.712 & 0.562\\
		& 3         &   & 0.399 & \bf0.474 & 0.418 & & 0.649 & \bf0.712 & 0.562\\[1ex]
		GA(2)  & 0.5 &   & \bf0.252 & 0.233 & 0.251 & & \bf0.402 & 0.360 & 0.343\\
		& 2   & & 0.188 & 0.233 & \bf0.251 & & 0.318 & \bf0.360 & 0.343\\
		& 3   & & 0.175 & 0.233 & \bf0.251 & & 0.269 & \bf0.360 & 0.343\\[1ex]
		GA(0.5)& 0.5 &  & 0.795 & \bf0.798 & 0.660 & & \bf0.957 & 0.952 & 0.821\\
		& 2   & & 0.764 & \bf0.798 & 0.660 & & 0.948 & \bf0.952 & 0.821\\
		& 3   & & 0.756 & \bf0.798 & 0.660 & & 0.942 & \bf0.952 & 0.821\\[1ex]
		U(0,1)    & 0.5 & & 0.089 & \bf0.134 & 0.073& & 0.180 & \bf0.236 & 0.074\\
		& 2   & & \bf0.143 & 0.134 & 0.073 & & \bf0.260 & 0.236 & 0.074\\
		& 3   & & \bf0.154 & 0.134 & 0.073 & & \bf0.258 & 0.236 & 0.074\\[1ex]
		Beta(2,1) & 0.5 & & 0.160 & \bf0.171 & 0.023 & & 0.270 & \bf0.296 & 0.024\\
		& 2   & & 0.170 & \bf0.171 & 0.023 & & \bf0.302 & 0.296 & 0.024\\
		& 3   & & \bf0.176 & 0.171 & 0.023& & \bf0.298 & 0.296 & 0.024\\\hline
	\end{tabular}
\end{table}

\subsection{Test for Normality} Suppose $X_1,\cdots,X_n$ is a random sample and we want to test whether the data come from N($\mu,\sigma^2$) population. For normality test, we replace $\mu$ and $\sigma^2$ in $\widehat{T}_{\alpha}(F,F_{\hat{\theta}})$ with their maximum likelihood estimates (MLE). The MLE of $\mu$ and $\sigma^2$ are $\bar{X}$ and $\frac{1}{n}\sum_{i=1}^{n}(X_i-\bar{X})^2$, respectively. Note that the statistic $\widehat{T}_{\alpha}(F,F_{\hat{\theta}})$ is scale invariant. We simulate 10000 random samples of size $n$ = 10, 15, 20 and 25 from N(0,1) distribution and calculate the critical points at 5\% level of significance. We take $\alpha$ = 0.5, 2 and 3 and compare the power of the tests with Kolmogorov-Smirnov (KS) test and KL divergence ($\alpha\to1$) -based tests. For power comparison, we consider Cauchy, exponential, gamma, uniform and beta distributions. The obtained results are presented in Tables \ref{tab3} and \ref{tab4}. From these tables, we observe that the proposed test performs better than KL and KS tests for Cauchy, exponential and gamma distributions when $\alpha$ = 0.5. For uniform and beta distributions, the proposed test performs better when $\alpha$ = 2 and 3. Note that the KL test is a special case of the proposed test. These two tests, in combination, outperform KS test except for Cauchy distribution when $\alpha$ = 2 and 3. The proposed test depends on the choice of $m$, and when the sample size increases, the power of all the tests also increase and the tests do attain the nominal significance level.

\begin{table}[h!]
	\centering
	\caption{Power of the tests for N(0,1) distribution.}\label{tab4}
	\begin{tabular}{|c| c| c| c c c| c| c c c|}
		\hline
		Alternatives & $\alpha$ & $n(m)$ & $\widehat{T}_{\alpha}(F,F_{\theta})$ & $KL(F,F_{\hat{\theta}})$ & KS & $n(m)$ & $\widehat{T}_{\alpha}(F,F_{\theta})$ & $KL(F,F_{\hat{\theta}})$ & KS\\
		\hline
		N(0,1)& 0.5 & 20(4)   & 0.049 & 0.050 & 0.050 & 25(5) & 0.052 & 0.052 & 0.049\\ 
		& 2   &     & 0.050 & 0.050 & 0.050 & & 0.050 & 0.052 & 0.049\\
		& 3   &     & 0.049 & 0.050 & 0.050 & & 0.051 & 0.052 & 0.049\\[1ex]
		C(0,1)& 0.5 &   & 0.758 & 0.700 & \bf0.814 & & 0.842 & 0.796 & \bf0.886\\
		& 2   &   & 0.633 & 0.700 & \bf0.814 & & 0.737 & 0.796 & \bf0.886\\
		& 3   &   & 0.572 & 0.700 & \bf0.814 & & 0.690 & 0.796 & \bf0.886\\[1ex]
		Exp(1)& 0.5 &  & \bf0.871 & 0.833 & 0.682 & & \bf0.945 & 0.931 & 0.776\\
		& 2   &   & 0.800 & \bf0.833 & 0.682 & & 0.908 & \bf0.931 & 0.776\\
		& 3         &   & 0.749 & \bf0.833 & 0.682 & & 0.875 & \bf0.931 & 0.776\\[1ex]
		GA(2)  & 0.5 &   & \bf0.507 & 0.441 & 0.426 & & \bf0.648 & 0.568 & 0.500\\
		& 2   & & 0.375 & \bf0.441 & 0.426 & & 0.500 & \bf0.568 & 0.500\\
		& 3   & & 0.311 & \bf0.441 & 0.426 & & 0.423 & \bf0.568 & 0.500\\[1ex]
		GA(0.5)& 0.5 &  & \bf0.993 & 0.991 & 0.914 & & \bf0.999 & 0.998 & 0.960\\
		& 2   & & 0.989 & \bf0.991 & 0.914 & & \bf0.998 & \bf0.998 & 0.960\\
		& 3   & & 0.984 & \bf0.991 & 0.914 & & 0.997 & \bf0.998 & 0.960\\[1ex]
		U(0,1)    & 0.5 & & 0.353 & \bf0.389 & 0.114 & & 0.471 & \bf0.502 & 0.133\\
		& 2   & & \bf0.419 & 0.389 & 0.114 & & \bf0.544 & 0.502 & 0.133\\
		& 3   & & \bf0.400 & 0.389 & 0.114 & & \bf0.520 & 0.502 & 0.133\\[1ex]
		Beta(2,1) & 0.5 & & \bf0.416 & 0.410 & 0.058 & & \bf0.556 & 0.540 & 0.088\\
		& 2   & & \bf0.431 & 0.410 & 0.058 & & \bf0.555 & 0.540 & 0.088\\
		& 3   & & 0.399 & \bf0.410 & 0.058 & & 0.521 & \bf0.540 & 0.088\\\hline
	\end{tabular}
\end{table}

\subsection{Test for Exponentiality} Let $X_1,\cdots,X_n$ be a random sample and we want to test whether this sample comes from Exp($\lambda$) distribution. We use the statistic $T_{\alpha}(F,F_{\hat{\theta}})$ and we replace the unknown parameter of the distribution by its MLE $\hat{\lambda}=\frac{1}{\bar{X}}$. We simulate 10000 random samples of size $n$ = 10, 15, 20 and 25 from Exp(1) population and calculate the critical points at 5\% level of significance. We choose $\alpha$ = 2 and $m=[\sqrt{n}+0.5]$ and compare the power of the tests with two popular tests based on entropy and cumulative residual entropy (Rao et al., 2004). The entropy-based test, introduced by Ebrahimi et al. (1992), has the test statistic as $$KL_{mn}=\exp(H_{mn}-\log\bar{X}-1),$$ where $H_{mn}$ is the Vasicek's entropy estimator. The null hypothesis is rejected if $KL_{mn}<C_1$, where $C_1$ is obtained from the empirical distribution of $KL_{mn}$. Baratpour and Rad (2012) proposed cumulative residual entropy-based test for exponentiality with the test statistic $$T=\frac{\sum_{i=1}^{n-1}\frac{n-i}{n}\log\left(\frac{n-i}{n}\right)(X_{(i+1):n}-X_{i:n})+\frac{\sum_{i=1}^{n}X_i^2}{2\sum_{i=1}^{n}X_i}}{\frac{\sum_{i=1}^{n}X_i^2}{2\sum_{i=1}^{n}X_i}}.$$
The null hypothesis is rejected when the calculated value of $T>C_2$, where the critical point $C_2$ is obtained from the null distribution of $T$. For reference distribution, we consider Weibull (WE), gamma (GA) and log-normal (LN) distributions having monotone decreasing, increasing and non-monotone hazard functions. For monotone decreasing hazard alternatives, we take WE(0.5), GA(0.4) and LN(2). For monotone increasing hazard alternatives, we consider WE(2), GA(2) and GA(3), while for non-monotone hazard alternatives, we consider LN(0.6) and LN(1.2). The power of the tests so obtained are provided in Tables \ref{tab5} and \ref{tab6}. From these tables, we find that proposed test performs better than the other two tests for monotone decreasing hazard alternative distributions.

\begin{table}[h]
	\centering
	\caption{Power of the tests for Exp(1) distribution.}\label{tab5}
	
	\begin{tabular}{|c| c |c c c |c| c| c c c|}
		\hline
		$n(m)$ & Alternatives & $\widehat{T}_{\alpha}(F,F_{\theta})$ & $KL_{m,n}$ & $T$ & $n(m)$ & Alternatives & $\widehat{T}_{\alpha}(F,F_{\theta})$ & $KL_{m,n}$ & $T$\\[1ex]
		\hline
		10(3) & WE(0.5) & \bf0.302 & 0.113 & 0.175  & 15(4) & WE(0.5) & \bf0.543 & 0.245 & 0.378\\
		& WE(2)   & 0.460 & \bf0.716 & 0.657 & & WE(2)   & 0.692 & \bf0.870 & 0.827\\
		& GA(0.4)   & \bf0.206 & 0.053 & 0.048 & & GA(0.4) & \bf0.360 & 0.116 & 0.161\\
		& GA(2) & 0.188 & \bf0.349 & 0.268 & & GA(2) & 0.286 & \bf0.447 & 0.332\\
		& GA(3) & 0.387 & \bf0.653 & 0.508 & & GA(3) & 0.610 & \bf0.803 &0.660\\
		& LN(0.6)  & 0.403 & \bf0.672 & 0.432 & & LN(0.6)  & 0.640 & \bf0.811 & 0.508\\
		& LN(1.2)  & \bf0.086 & 0.052 & 0.077 & & LN(1.2)  & 0.123 & 0.060  & \bf0.154\\
		& LN(2)    & \bf0.386 & 0.204 & 0.306 & & LN(2)    & \bf0.608 & 0.386 & 0.549\\
		& Exp(1) & 0.047 & 5.37 & 4.72 & & Exp(1)        & 0.048  & 5.09  & 4.82\\
		\hline	
	\end{tabular}
\end{table}
\begin{table}[h!]
	\centering
	\caption{Power of the tests for Exp(1) distribution.}\label{tab6}
	
	\begin{tabular}{|c| c |c c c |c| c| c c c|}
		\hline
		$n(m)$ & Alternatives & $\widehat{T}_{\alpha}(F,F_{\theta})$ & $KL_{m,n}$ & $T$ & $n(m)$ & Alternatives & $\widehat{T}_{\alpha}(F,F_{\theta})$ & $KL_{m,n}$ & $T$\\[1ex]
		\hline
		20(4) & WE(0.5) & \bf0.747 & 0.548 & 0.542  & 25(5) & WE(0.5) & \bf0.850 & 0.631 & 0.664\\
		& WE(2)   & 0.764 & \bf0.937 & 0.917 & & WE(2)   & 0.900 & \bf0.972 & 0.968\\
		& GA(0.4)   & \bf0.607 & 0.339 & 0.250 & & GA(0.4) & \bf0.695 & 0.391 & 0.330\\
		& GA(2) & 0.268 & \bf0.524 & 0.389 & & GA(2) & 0.355 & \bf0.588 & 0.446\\
		& GA(3) & 0.665 & \bf0.893 & 0.736 & & GA(3) & 0.812 & \bf0.949 & 0.831\\
		& LN(0.6)  & 0.700 & \bf0.910 & 0.576 & & LN(0.6)  & 0.852 & \bf0.954 & 0.646\\
		& LN(1.2)  & 0.162 & 0.122 & \bf0.215 & & LN(1.2)  & 0.191 & 0.123  & \bf0.291\\
		& LN(2)    & \bf0.774 & 0.660 & 0.720 & & LN(2)    & \bf0.876 & 0.743 & 0.825\\
		& Exp(1) & 0.051 & 0.054 & 0.047 & & Exp(1)        & 0.051  & 0.050  & 0.048\\
		\hline	
	\end{tabular}
\end{table}

\subsection{Goodness-of-fit Tests for Progressive Type-II Samples} Goodness-of-fit tests for certain distributions under PC-II censored samples can be developed using Tsallis divergence. Suppose $n$ identical units are placed on a lifetest, each having a common lifetime distribution with pdf $f_{\theta}$ and cdf $F_{\theta}$. Let $X_{1:r:n},\cdots,X_{r:r:n}$ be a PC-II censored sample with censoring scheme $(R_1,R_2,\cdots,R_r)$. Based on this sample, we want to test the hypothesis whether the data come from $f_{\theta}$ or not. Let $F_{r:n}(x)$ denote the empirical distribution function of the PC-II censored data. Balakrishnan and Sandhu (1995) expressed $F_{r:n}(x)$ as

$$F_{r:n}(x)=
\begin{cases}
0,\;\;\;\;\mbox{if}\;x<X_{1:r:n},\\
\alpha_{i:r:n},\;\;\;\mbox{if}\;X_{i:r:n}\leq x<X_{i+1:r:n},\;i=1,2,\cdots r-1,\\
\alpha_{r:r:n},\;\;\;\mbox{if}\;x \geq X_{r:r:n},
\end{cases}$$
where $\alpha_{i:r:n}=E(U_{i:r:n})$. Now empirical Tsallis divergence between $F_{r,n}$ and $F$ under PC-II sample can be defined, analogous to (\ref{e18}), as 

\begin{eqnarray}\label{e28}
\hat{T_{\alpha}}(F_{r,n},F)&=&\frac{1}{\alpha-1}\left( \frac{1}{r}\sum_{i=1}^{r}\left(\frac{U_{i+m:r:n}-U_{i-m:r:n}}{F_{\hat{\theta}}(X_{i+m:r:n})-F_{\hat{\theta}}(X_{i-m:r:n})} \right)^{\alpha-1} -1 \right)\nonumber\\
&\approx&\frac{1}{\alpha-1}\left( \frac{1}{r}\sum_{i=1}^{r}\left(\frac{ \prod_{k=r-(i-m)+1}^{r}\beta_k - \prod_{k=r-(i+m)+1}^{r}\beta_k }{F_{\hat{\theta}}(X_{i+m:r:n})-F_{\hat{\theta}}(X_{i-m:r:n})} \right)^{\alpha-1} -1 \right),
\end{eqnarray}

\noindent where $\hat{\theta}$ is an estimator of $\theta$ under PC-II setup. Note that, $\hat{T_{\alpha}}(F_{r,n},F)$ is a semi-parametric estimator for Tsallis divergence between empirical density and fitted parametric density under PC-II censored data. We can use $\hat{T_{\alpha}}(F_{r,n},F)$ as a test statistics for testing goodness-of-fit for PC-II censored data.

\subsubsection{Exponentiality under Progressive Type-II Samples} For exponentiality testing, the test statistic is given by $$\hat{T_{\alpha}}(F_{r,n},F)=\frac{1}{\alpha-1}\left( \frac{1}{r}\sum_{i=1}^{r}\left(\frac{ \prod_{k=r-(i-m)+1}^{r}\beta_k - \prod_{k=r-(i+m)+1}^{r}\beta_k }{\exp(-\hat{\theta}X_{i-m:r:n})-\exp(-\hat{\theta}X_{i+m:r:n})} \right)^{\alpha-1} -1 \right),$$
where $\hat{\theta}=\frac{r}{\sum_{i=1}^{r}(1+R_i)X_{i:r:n}}$ is the MLE for $\theta$ under PC-II censoring. We will reject exponentiality hypothesis if $\hat{T_{\alpha}}(F_{r,n},F)$ is large. We simulate 5000 random samples from Exp(1) distributions under various PC-II censoring schemes and obtain critical points at 10\% significance level. Entropy based exponentiality tests under PC-II censoring perform better when the alternative distribution has increasing hazard function, see Balakrishnan et al. (2007). So, we consider WE(2) and GA(2) distributions for power calculation and compare the performance of the proposed test with the test proposed by Balakrishnan et al. (2007). Let us denote it by $T_{PC}$. We choose $\alpha$ = 2 and $m=[\sqrt{r}+0.5]$ for power computation and provide the results in Table \ref{tab7}. From the table we have observed that neither test outperforms the other for all the cases. Performance of both tests increases when sample size increases. Both tests perform better for WE(2) than for GA(2) alternative. The proposed test does attain the nominal level of significance. We mostly considered one step censoring schemes for illustrations. For other censoring schemes, the proposed tests outperforms $T_{PC}$ for both the alternative distributions.

\begin{table}[h]
	\centering
	\caption{Power of exponentiality tests under PC-II samples at 10\% significance level.}\label{tab7}
	\begin{tabular}{|c| c |c |c c c| c c|}
		\hline
		& & & \multicolumn{3}{c}{$\hat{T_{\alpha}}(F_{r,n},F)$}&\multicolumn{2}{c|}{$T_{PC}$}\\
		\hline
		$n$ & $r$ & Schemes & Exp(1) & WE(2) & GA(2) & WE(2) & GA(2)\\
		\hline
		20 & 10 & ($10,0*9$) & 0.0952 & 0.8054 & 0.3518 & \bf0.9280 & \bf0.6110\\
		& & ($0,10,0*8$) & 0.0984 & 0.8044 & 0.3478 & \bf0.9382 & \bf0.6735\\
		& & ($0*8,10,0$) & 0.1050 & 0.5098 & 0.2580 & \bf0.5815 & \bf0.3803\\
		& & ($0*9,10$)   & 0.1044 & 0.6650 & 0.4404 & \bf0.7295 & \bf0.5277\\
		& & ($1*10$) & 0.0972 & \bf0.9734 & \bf0.6684 & 0.8650 & 0.6044\\\hline
		20 & 15 & ($5,0*14$) & 0.1050 & \bf0.9746 & 0.6210 & 0.9711 & \bf0.6789\\
		& & ($0,5,0*13$) & 0.099 & 0.9698 & 0.5902 & \bf0.9755 & \bf0.6982\\
		& & ($0*13,5,0$) & 0.1000 & \bf0.8764 & \bf0.5234 & 0.8236 & 0.4966\\
		& & ($0*14,5$)   & 0.0948 & 0.8750 & 0.5646 & \bf0.8900 & \bf0.6020\\
		& & ($1,1,0*5,1,0*5,1,1$) & 0.1050 & \bf0.9772 & \bf0.7106 & 0.9592 & 0.6970\\\hline
		30 & 20 & ($10,0*19$) & 0.0960 & \bf0.9894 & 0.6708 & 0.9870 & \bf0.7718\\
		& & ($0,10,0*18$) & 0.1062 & \bf0.9898 & 0.6600 & 0.9849 & \bf0.7913\\
		& & ($0*18,10,0$) & 0.1014 & \bf0.8866 & \bf0.5292 & 0.7570 & 0.4684\\
		& & (0*19,10) & 0.1044 & 0.9276 & 0.6428 & \bf0.9563 & \bf0.7189\\
		& & ($1,0,1,0,\cdots,1,0$) & 0.0918 & \bf0.9978 & \bf0.7866 & 0.9741 & 0.6877\\\hline 
		
	\end{tabular}
\end{table}

\subsection{Data Analyses} We analyse three real datasets to illustrate the performance of the proposed tests. The first dataset shows the effectiveness of the proposed normality test, while the last two data sets illustrate the proposed exponentiality test. Also, we analysed three PC-II censored data sets as well. 

\subsubsection*{Data Set 1} This data set, given in Gun et al. (2008), gives the weights at birth (in kg.) of 15 babies born in a hospital in Calcutta. The observations are given as follows:
\begin{table}[h!]
	\centering
	\begin{tabular}{cccccccc}
	2.79, & 2.56, & 3.64, & 3.01, & 2.16, & 2.25, & 3.19, & 3.06,\\
	2.61, & 3.10, & 3.42, & 3.55, & 3.88, & 3.51, & 3.82.
		
	\end{tabular}
\end{table}
\\
 We now apply the test for normality. Here, the sample size is $n$ = 15 so that $m$ = 4. The sample mean and the sample standard deviation are found to be 3.07 and 0.4899, respectively. For $\alpha$ = 0.5, the value of the test statistic is 0.20127 and the p-value is 0.5557. For $\alpha$ = 2, the test statistic value is 0.40210 and the p-value is 0.3734. For $\alpha$ = 3, the test statistic value is 0.69992 and the p-value is 0.3492. So, the p-values are quite large in all three cases and so we can conclude that the dataset of babies' weights at birth is normally distributed.

\subsubsection*{Data Set 2} The following dataset provides the mileages for 19 military personnel
carriers that failed in service:
\begin{table}[h!]
	\centering
	\begin{tabular}{cccccccccc}
		162, & 200, & 271, & 320, & 393, & 508, & 539, & 629, & 706, & 778, \\
		884, & 1003, & 1101, & 1182, & 1463, & 1603, & 1984, & 2355, & 2880. &
	\end{tabular}
\end{table}
\\
 This dataset was considered by Grubbs (1971), who found that the exponential distribution fits this dataset well. Now, we apply our exponentiality test to these data. Here, $n$ = 19 so that $m$ = 4 (based on the formula $m=[\sqrt{n}+0.5]$). We take $\alpha$ = 2. The MLE for the exponential parameter is $\frac{1}{\bar{X}}=0.001$. The test statistic value is 0.43802 and the p-value is 0.3171, which supports the hypothesis that these data follow the exponential distribution.

\subsubsection*{Data Set 3} We consider a dataset that represents quantity of thousands of cycles to failure for electrical appliances; see Lawless (2011). The sample data is as follows:
\begin{table}[h!]
	\centering
	\begin{tabular}{cccccccccc}
		0.014, & 0.034, & 0.059, & 0.061, & 0.069, & 0.080, & 0.123, & 0.142, & 0.165, & 0.210,\\
		0.381, & 0.464, & 0.479, & 0.556, & 0.574, & 0.839, & 0.917, & 0.969, & 0.991, & 1.064,\\
		1.088, & 1.091, & 1.174, & 1.270, & 1.275, & 1.355, & 1.397, & 1.477, & 1.578, & 1.649,\\
		1.702, & 1.893, & 1.932, & 2.001, & 2.161, & 2.292, & 2.326, & 2.337, & 2.628, & 2.785,\\
		2.811, & 2.886, & 2.993, & 3.122, & 3.248, & 3.715, & 3.790, & 3.857, & 3.912, & 4.100.
	\end{tabular}
\end{table}
\\
Yousaf et al. (2019) found that Chen distribution fits this dataset better than the exponential
distribution. Chen distribution has its distribution function as $$F(x;\eta,\lambda)=1-\exp\left(\eta\left(1-\exp( x^{\lambda}) \right) \right),\;x>0,\;0<\eta<\infty,\;\lambda>0.$$ Recently, Xiong et al. (2022) applied many popular tests like
Kolmogorov-Smirnov, Kuiper, Cramer-von Mises and Anderson-Darling, and observed that they all fail to reject the null hypothesis that the data are from exponential distribution, at 5\% level of significance. We take $\alpha$ = 2 and since $n$ = 50, and choose $m$ = 7. The MLE for the exponential parameter is $0.64073$ and the value of the test statistic is 0.41352. The p-value is 0.0144, and so we reject the null hypothesis of exponentiality, at 5\% level of significance.

\begin{table}[h]
	\centering
	\caption{Progressive censored sample generated from data set 2 with ($n, m$) = (19, 10) and $R=(9,0*9)$.}\label{tab8}
	\resizebox{14cm}{!}{
		\begin{tabular}{c c c c c c c c c c c}
			\hline	
			$i$      &  1 &   2   & 3 & 4 & 5 & 6 & 7 & 8 & 9 & 10 \\\hline
			$x_{i:10:19}$ & 162 & 200 & 393 & 508 & 539 & 778 & 884 & 1003 & 1463 & 1984 \\\hline
			$R_i$ & 9 & 0 & 0 & 0 & 0 & 0 & 0 & 0 & 0 & 0\\\hline
	\end{tabular}}	
\end{table}

\begin{table}[h]
	\centering
	\caption{Progressive censored sample generated from data set 2 with ($n, m$) = (19, 10) and $R=(0*9,9)$.}\label{tab9}
	\resizebox{14cm}{!}{
		\begin{tabular}{c c c c c c c c c c c}
			\hline	
			$i$      &  1 &   2   & 3 & 4 & 5 & 6 & 7 & 8 & 9 & 10 \\\hline
			$x_{i:10:19}$ & 162 & 200 & 271 & 393 & 508 & 539 & 629 & 706 & 778 & 884 \\\hline
			$R_i$ & 0 & 0 & 0 & 0 & 0 & 0 & 0 & 0 & 0 & 9\\\hline
	\end{tabular}}	
\end{table}

\begin{table}[h!]
	\centering
	\caption{Progressive censored sample generated from data set 2 with ($n, m$) = (19, 10) and $R=(5,0*8,4)$.}\label{tab10}
	\resizebox{14cm}{!}{
		\begin{tabular}{c c c c c c c c c c c}
			\hline	
			$i$      &  1 &   2   & 3 & 4 & 5 & 6 & 7 & 8 & 9 & 10 \\\hline
			$x_{i:10:19}$ & 162 & 200 & 271 & 320 & 393 & 539 & 706 & 778 & 1003 & 1182 \\\hline
			$R_i$ & 5 & 0 & 0 & 0 & 0 & 0 & 0 & 0 & 0 & 4\\\hline
	\end{tabular}}	
\end{table}
\subsubsection*{Progressively Type-II Censored Data Sets} From data set 2 of Grubbs et al. (1971), we generate three PC-II censored data set and apply the proposed exponentiality test to these three data. The data dets along with censoring schemes are provided in Tables \ref{tab8}-\ref{tab10}, respectively. We calculate the test statistics and corresponding p-values for these three data sets. 

\begin{itemize}
	\item For the 1st data set, $\hat{T_{\alpha}}(F_{r,n},F)$ = 0.20322 and p-value is 0.617 and hence, we cannot reject the null hypothesis which is the data follows exponential distribution.
	
	\item For the 2nd data, $\hat{T_{\alpha}}(F_{r,n},F)$ = 0.30578 and p-value is 0.213 which is in favour of the null hypothesis.
	
	\item For the 3rd data, $\hat{T_{\alpha}}(F_{r,n},F)$ = 0.41095 and p-value is 0.323, suggesting acceptance of exponentiality hypothesis.
\end{itemize}

\section{Concluding remarks}\label{S9} In this paper, we have introduced four estimators for Tsallis entropy measure and have compared their performance. We have also discussed the estimation of Tsallis divergence measure and have developed goodness-of-fit tests for normal and exponential distributions based on it. We have compared the power of the test with various other tests using simulations and have shown that the proposed tests perform well. We have analysed three real datasets to show the effectiveness of the proposed tests. For comparing the proposed estimators, we have used two choices of the window size $m$ as $m=[\sqrt{n}+0.5]$ and $m=\left[ \frac{n}{3}\right] $, where $n$ is the sample size. Other choices of $m$ could also be explored. We only used one estimator for Tsallis divergence and it is used along with $m=[\sqrt{n}+0.5]$ to develop goodness-of-fit tests. Other estimators can also be used for this purpose. We have introduced an estimator for Tsallis entropy based on progressively type-II censored data and have calculated its mean and variance for various progressive type-II censoring schemes using simulations. From the obtained results, we have observed that the variance of the estimator decreases as sample size increases. We have further proposed Tsallis divergence estimation for progressive type-II censored data and developed an exponentiality test when the data is progressively type-II censored. We have studied the power of the test for monotone increasing hazard alternatives using simulation and found that the test performs well. We have considered $\alpha$ = 2, $m=[\sqrt{r}+0.5]$ and a few censoring schemes for power comparisons. Other choices of $\alpha$ and $m$ can be considered as well. Also, we have proposed a quantile density function based estimator for Tsallis entropy, studied its asymptotic properties, and evaluated its performance for normal, exponential and Govindarajulu distributions. This estimator uses a kernel density estimator and we took the bandwidth using normal reference rule as $h_n=1.06\sigma n^{-\frac{1}{5}}$. For non-normal data, bias and MSE of the estimator can also be compared for different choices of $h_n$. We have proposed another quantile based estimator for Tsallis entropy. Performance of this estimator can be evaluated in a future study. Optimal choice for bandwidth parameter for both these estimators under non-normal data could be an interesting problem for future research.

\section*{Conflicts of Interest} The authors declare no conflict of interest.

\section*{Funding} No funding is received for this work.

\end{document}